\documentclass[runningheads]{llncs}
\usepackage{graphicx}
\usepackage{tikz}	
\usepackage{color, colortbl}
\usepackage{booktabs}
\usepackage{pdflscape}
\usepackage{xspace}
\usepackage{multicol}
\usepackage{color,soul}

\let\proof\relax
\let\endproof\relax

\usepackage{amsmath}
\usepackage{listings}
\usepackage{amssymb}
\usepackage{comment}
\usepackage{multirow}
\usepackage{mfirstuc}
\usepackage{textcase}
\usepackage{algorithm} 
\usepackage[noend]{algpseudocode} 
\usepackage{hyperref}
\usepackage{enumerate}
\usepackage{wasysym}
\usepackage{varwidth}
\usepackage{booktabs}
\usetikzlibrary{automata,arrows,positioning,shapes.geometric,decorations.pathmorphing}
\usetikzlibrary{arrows.meta}

\usepackage{amsthm}
\algnewcommand{\IIf}[1]{\State\algorithmicif\ #1\ \algorithmicthen}
\algnewcommand{\EndIIf}{\unskip\ \algorithmicend\ \algorithmicif}

\newcommand{\nf}[1]{#1}
\newcommand{\blue}[1]{\textcolor{blue}{#1}}

\newcommand{\bound}{vb}

\newcommand{\variables}{\emph{Var}}
\newcommand{\underapproxFunc}{\textsc{UnderApprox}}

\newcounter{magicrownumbers}

\newtheoremstyle{remboldstyle}
  {}{}{\itshape}{}{\bfseries}{.}{.5em}{{\thmname{#1 }}{\thmnumber{#2}}{\thmnote{ (#3)}}}
\theoremstyle{remboldstyle}

\makeatletter
\algrenewcommand\ALG@beginalgorithmic{\scriptsize}
\makeatother

\theoremstyle{definition}

\makeatletter
\let\c@example\relax
\makeatother

\newtheorem{example}{Example}

\newtheorem{named}{Axiom}
\newcommand{\class}{\textsc{cls}}

\newcommand{\op}[1]{\textbf{\MakeTextUppercase{#1}}}
\newcommand{\always}[2]{\square_{#2}~#1}
\newcommand{\opnext}[2]{\Circle_{#2}~#1}
\newcommand{\eventually}[2]{\lozenge_{#2}~#1}
\newcommand{\once}[2]{\blacklozenge_{#2}~#1}
\newcommand{\until}[3]{#1~\mathcal{U}_{#3}~#2}

\newcommand{\prev}[2]{\CIRCLE_{#2}~#1}
\newcommand{\since}[3]{#1~\mathcal{S}_{#3}~#2}

\newcommand{\tfol}{FOL$^*$\xspace}
\newcommand{\tp}{\textsc{TP}}
\newcommand{\domain}{D}
\newcommand{\domainUnder}{\domain_{\downarrow}}

\newcommand{\reqs}{\textit{Reqs}}
\newcommand{\reqsUnder}{\textit{Reqs}_{\downarrow}}
\newcommand{\translate}{T}
\newcommand{\mainAlg}{\textsc{IBS}\xspace}
\newcommand{\opt}{\text{op}}
\newcommand{\noopt}{\text{nop}}
\newcommand{\naiveAlg}{\textsc{NBS}}
\newcommand{\groundAlg}{G}
\newcommand{\mfotl}{MFOTL\xspace}

\newcommand{\dataDomainConstraint}{T_{data}}
\newcommand{\negp}{\neg P}
\newcommand{\negpf}{\neg P_f}
\newcommand{\groundinput}{\phi_f}
\newcommand{\underinput}{\phi_\downarrow}

\newenvironment{sop}{%
  \proof}{\endproof}

\definecolor{mygray}{rgb}{0.5,0.5,0.5}
\definecolor{mymauve}{rgb}{0.1,0.2,0.7}
\definecolor{olivegreen}{cmyk}{.6,.4,0.8,0}

\makeatletter
\algrenewcommand\ALG@beginalgorithmic{\scriptsize}
\makeatother

\makeatletter
\newcommand\footnoteref[1]{\protected@xdef\@thefnmark{\ref{#1}}\@footnotemark}
\makeatother

\begin{document}

\title{Early Verification of Legal Compliance via Bounded Satisfiability Checking}
%
%
\author{}
\institute{}
\author{Nick Feng\inst{1} \and Lina Marsso\inst{1} \and Mehrdad Sabetzadeh\inst{2} \and Marsha Chechik\inst{1}}
%
\authorrunning{N. Feng et al.}
%
\institute{University of Toronto, Canada. 
\email{\{fengnick,lmarsso,chechik\}@cs.toronto.edu} \and University of Ottawa, Canada. \email{m.sabetzadeh@uottawa.ca}}
\newpage

\maketitle              
\begin{abstract}
\vspace{-0.3in}
Legal properties involve reasoning about data values and time.
Metric first-order temporal logic (MFOTL) provides a rich formalism for specifying legal properties. While MFOTL has been successfully used for verifying legal properties over operational systems via runtime monitoring, no solution exists for MFOTL-based verification in early-stage system development captured by requirements.
Given a legal property and system requirements, both formalized in MFOTL, the compliance of the property can be verified on the requirements via satisfiability checking. In this paper, we propose a practical, sound, and complete (within a given bound)  satisfiability checking approach for MFOTL.
The approach, based on satisfiability modulo theories (SMT), employs a counterexample-guided strategy to incrementally  search for a satisfying solution. We implemented our approach using the Z3 SMT solver and evaluated it on five case studies spanning the healthcare, business administration, banking and aviation domains.  Our results indicate that our approach can efficiently determine whether legal properties of interest are met, or generate counterexamples that lead to compliance violations.

\end{abstract}
\section{Introduction}
\label{sec:intro}
Software systems, such as medical systems, are increasingly required to comply with laws and regulations aimed at ensuring safety, security, and data privacy~\cite{Arfelt-Basin-Debois-19,Shan-et-al-19}. 
The properties stipulated by these laws and regulations -- which we  refer to as \emph{legal properties} (LP) hereafter -- typically involve reasoning about actions, ordering and time. 
As an example, consider the following LP, $P1$, derived from a health-data regulation (s.11, PHIPA~\cite{phipa-2004}): ``If personal health information is not accurate or not up-to-date, it should not be accessed''. 
In this property, the accuracy and the freshness of the data depend on how and when the data was collected and updated before being accessed. 
Specifically, this property constrains the data action \textit{access} to have accurate and up-to-date data values, which further constrains the order and time of $\textit{access}$ with respect to other data actions.

System compliance with LPs can be checked on the system design 
or on an operational model of a system implementation.
In this paper, we focus on the early stage, where one can check whether a formalization of the system requirements satisfies an LP.  The formalization can be done using a descriptive formalism like temporal logic
~\cite{Rozier-et-al-07,Li-et-al-19}.
For instance, the requirement (req$_0$) of a data collection system: ``no data can be accessed prior to 15 days after the data has been collected''  needs to be formalized for verifying compliance of $P1$. 
It is important to formalize the data and time constraints of both the system requirements and LPs, such as the ones of $P1$ and req$_0$.

\emph{Metric first-order temporal logic (MFOTL)} enables the specification of data and time constraints~\cite{Basin-10} and has an expressive formalism for capturing LPs and the related system requirements that constrain data and time~\cite{Arfelt-Basin-Debois-19}.
Existing work on \mfotl verification focuses on detecting violations at run-time through monitoring~\cite{Arfelt-Basin-Debois-19,DBLP:conf/esorics/HubletBK22}, with \mfotl formulas being checked on execution logs. 
There is an unsatisfied need for determining the \emph{satisfiability} of \mfotl specifications, i.e., looking for LP violations possible in \mfotl specification. 
\nf{This is important for designing system requirements that comply with LPs.}


\mfotl satisfiability checking is generally undecidable since \mfotl is an extension of first-order logic (FOL).
Restrictions are thus necessary for making the problem decidable. 
In this paper, we restrict ourselves to safety properties. 
For safety properties, LP violations are finite sequences of data actions, captured via a finite-length counterexample.
For example, a possible violation of $P1$ is a sequence consisting of storing a value $v$ in a  variable $d$,  updating $d$'s value to $v'$, then reading $d$ again and not obtaining $v'$.
Since we are interested in finite counterexamples,  bounded verification is a natural strategy to pursue for achieving decidability.
SAT solvers have been previously used for bounded satisfiability checking of metric temporal logic (MTL)~\cite{Rozier-et-al-07,Li-et-al-19}.
However, MTL cannot effectively capture quantified data constraints 
in LPs, hence  the solution is not applicable directly.  As an extension to MTL, MFOTL can effectively capture data constraints used in LP.
Yet, to the best of our knowledge, there has not been any  prior work on bounded \mfotl  satisfiability checking. 

To establish a \emph{bound} in bounded verification, researchers have 
predominantly
relied on bounding the \emph{size of the universe}~\cite{Garavel-Graf-13}. 
Bounding the universe would be too restrictive because LPs routinely refer to variables with large ranges, e.g., timed actions spanning several years.
Instead, we bound the \emph{number of data actions in a run}, which bounds the number of actions in the counterexample.  

Equipped with our proposed notion of a bound, we develop an incremental approach ($\mainAlg$) for bounded satisfiability checking of MFOTL.  We first translate the \mfotl property and requirements into first-order logic formulas with quantified relational objects (\tfol).
We then incrementally ground the \tfol constraints to eliminate the quantifiers by considering an increasing number of relational objects.
Subsequently, we check the satisfiability of the resulting constraints using an SMT solver. 
Specifically, we make the following contributions:  (1) we propose a translation of \mfotl formulas to \tfol; (2) we provide a novel bounded satisfiability checking solution, \nf{$\mainAlg$}, for the translated \tfol formulas with incremental and counterexample-guided over/ under-approximation. 
\nf{Note that while our solution to \mfotl satisfibility checking can be applied to a broader domain of applications, in this paper we focus on the legal domain. 
} 
We empirically evaluate  $\mainAlg$ on five case studies with a total of 24 properties showing that it can effectively and efficiently 
find LP violations or prove satisfiability.

The rest of this paper is organized as follows.
Sec.~\ref{sec:background} provides background and establishes our notation.
Sec.~\ref{sec:bsc} defines the bounded satisfiability checking (BSC) problem.
Sec.~\ref{sec:approach} provides an overview of our solution and the translation of \mfotl to \tfol.
Sec.~\ref{sec:incremental} presents our solution, and  proofs of its soundness, termination and optimality are in
Sec.~\ref{ap:SoundTerminateOpt}.
Sec.~\ref{sec:evaluation} reports on the experiments performed to validate our bounded satisfiability checking solution for \mfotl.
Sec.~\ref{sec:relatedwork} discusses related work.
\hbox{Sec.~\ref{sec:conclusion} concludes the paper.}

\section{Preliminaries}
\label{sec:background}
\vspace*{-0.9em}
In this section, we describe metric first-order temporal logic (MFOTL)~\cite{Basin-10}. 

\begin{figure}[t]
    \centering
    \scalebox{.78}{  
        \begin{tabular}{l}
            $P1 = \always{\forall d, v (\textit{Access}(d,v)) \implies \since{( \forall v'( v' \neq v \Rightarrow \neg \textit{Update}(d, v') \wedge \neg \textit{Collect}(d,v')))}{(\textit{Update}(d, v) \vee \textit{Collect}(d, v))}{})}{}$\\
            If a personal health information is not accurate or not up-to-date, it should not be accessed.\\      
            \hline
            $req_0 = \always{\forall d, v  (\textit{Access}(d, v) \implies \once {\exists v' . \textit{Collect(d, v')}}{[360, )}}{}$\\
            No data is allowed to be accessed before the data ID has been collected for at least 15 days (360 hours).\\
            \hline
            $req_1 = \always{\forall d, v (\textit{Update}(d,v) \implies \neg (\once{\exists v' . (\textit{Collect}(d,v') \vee \textit{Update}(d,v'))}{[1, 168]}))}{}$ \\
            Data value can only be updated after having been collected or last updated for more than a week (168 hours).\\\hline
            $req_2 = \always{\forall d,v (\textit{Access}(d,v) \implies \once{\textit{Collect}(d,v) \vee \textit{Update}(d,v)}{[0, 168]})}{}$\\ 
            Data can only be accessed if has been collected or updated within a week (168 hours).\\\hline
            $req_3 = \always{\forall d, v (\textit{Collect}(d, v) \implies \neg (\exists v'' . (\textit{Collect}(d, v'') \wedge v \neq v'') \vee \once{\exists v' . \textit{Collect}(d, v'))}{[1,)})}{}$ No data re-collection.
        \end{tabular}}
    \caption{\small{Example requirements and legal property $P1$ of DCC,  with signature \protect\scriptsize$S_{data} =(\emptyset, \{\textit{Collect}, \; \textit{Update}, \; \textit{Access}\}, \iota_{data})$, {\footnotesize where}  $\iota_{data}(\textit{Collect}) = \iota_{data}(\textit{Update}) = \iota_{data}(\textit{Access}) = 2$.}}
    \label{fig:formulas}
\end{figure}


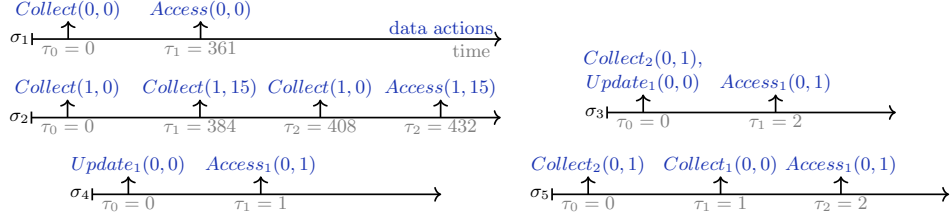
\begin{figure}[t]
  \centering
    \begin{tabular}{cc}
            \scalebox{.8}{
    \begin{tikzpicture}[scale=2]
      \node at (3.6,1.1) {\small\color{mymauve}data actions};
      \node at (3.85,0.9) {\small\color{mygray}time};
      
      \node at (0.5,1.24) {\small\color{mymauve}$Collect(0,0)$};
      \node at (1.6,1.24) {\small\color{mymauve}$Access (0,0)$};
      
      \node at (0.5,0.92) {\small\color{mygray}$\tau_0 =0$};
      \node at (1.6,0.92) {\small\color{mygray}$\tau_1 =361$};
      
      \node at (0.1,1.0) {$\sigma_1$};

      \draw [thick,black,->] (0.2,1) -- (4.1,1);
      \draw [thick,black,-] (0.2,1.05) -- (0.2,0.95);
      \draw [thick,black,->] (0.5,1) -- (0.5,1.15);
      \draw [thick,black,->] (1.6,1) -- (1.6,1.15);
      
    \end{tikzpicture}
    }
    &
    %
    \multirow{2}{*}{
    \scalebox{.8}{
        \begin{tikzpicture}[scale=2]
          
          \node at (0.5,1.44) {\small\color{mymauve}$Collect_2(0, 1)$,};
          \node at (0.5,1.24) {\small\color{mymauve}$Update_1(0, 0)$};
          \node at (1.6,1.24) {\small\color{mymauve}$Access_1(0, 1)$};
          
          \node at (0.5,0.92) {\small\color{mygray}$\tau_0 =0$};
          \node at (1.6,0.92) {\small\color{mygray}$\tau_1 =2$};
          
          \node at (0.1,1.0) {$\sigma_3$};
    
          \draw [thick,black,->] (0.2,1) -- (2.6,1);
          \draw [thick,black,-] (0.2,1.05) -- (0.2,0.95);
          \draw [thick,black,->] (0.5,1) -- (0.5,1.15);
          \draw [thick,black,->] (1.6,1) -- (1.6,1.15);
          
        \end{tikzpicture}
        }
    }
    \\
    %
    \scalebox{.8}{
    \begin{tikzpicture}[scale=2]
      
      \node at (0.5,1.24) {\small\color{mymauve}$Collect(1,0)$};
      \node at (1.6,1.24) {\small\color{mymauve}$Collect (1,15)$};
      \node at (2.6,1.24) {\small\color{mymauve}$Collect (1, 0)$};
      \node at (3.6,1.24) {\small\color{mymauve}$Access (1, 15)$};
      
      \node at (0.5,0.92) {\small\color{mygray}$\tau_0 =0$};
      \node at (1.6,0.92) {\small\color{mygray}$\tau_1 =384$};
      \node at (2.6,0.92) {\small\color{mygray}$\tau_2 =408$};
      \node at (3.6,0.92) {\small\color{mygray}$\tau_2 =432$};
      
      \node at (0.1,1.0) {$\sigma_2$};

      \draw [thick,black,->] (0.2,1) -- (4.1,1);
      \draw [thick,black,-] (0.2,1.05) -- (0.2,0.95);
      \draw [thick,black,->] (0.5,1) -- (0.5,1.15);
      \draw [thick,black,->] (1.6,1) -- (1.6,1.15);
      \draw [thick,black,->] (2.6,1) -- (2.6,1.15);
      \draw [thick,black,->] (3.6,1) -- (3.6,1.15);
      
    \end{tikzpicture}
    }

    &

    \\

        \scalebox{.8}{
    \begin{tikzpicture}[scale=2]
      
      \node at (0.5,1.24) {\small\color{mymauve}$Update_1(0,0)$};
      \node at (1.6,1.24) {\small\color{mymauve}$Access_1(0, 1)$};
      
      \node at (0.5,0.92) {\small\color{mygray}$\tau_0 =0$};
      \node at (1.6,0.92) {\small\color{mygray}$\tau_1 =1$};
      
      \node at (0.1,1.0) {$\sigma_4$};

      \draw [thick,black,->] (0.2,1) -- (3.1,1);
      \draw [thick,black,-] (0.2,1.05) -- (0.2,0.95);
      \draw [thick,black,->] (0.5,1) -- (0.5,1.15);
      \draw [thick,black,->] (1.6,1) -- (1.6,1.15);
      
    \end{tikzpicture}
    }
    &
    \scalebox{.8}{
    \begin{tikzpicture}[scale=2]
      
      \node at (0.5,1.24) {\small\color{mymauve}$Collect_2(0, 1)$};
      \node at (1.6,1.24) {\small\color{mymauve}$Collect_1(0, 0)$};
      \node at (2.6,1.24) {\small\color{mymauve}$Access_1(0, 1)$};
      
      \node at (0.5,0.92) {\small\color{mygray}$\tau_0 =0$};
      \node at (1.6,0.92) {\small\color{mygray}$\tau_1 =1$};
      \node at (2.6,0.92) {\small\color{mygray}$\tau_2 =2$};
      
      \node at (0.1,1.0) {$\sigma_5$};

      \draw [thick,black,->] (0.2,1) -- (3.5,1);
      \draw [thick,black,-] (0.2,1.05) -- (0.2,0.95);
      \draw [thick,black,->] (0.5,1) -- (0.5,1.15);
      \draw [thick,black,->] (1.6,1) -- (1.6,1.15);
      \draw [thick,black,->] (2.6,1) -- (2.6,1.15);
      
    \end{tikzpicture}
    }
    %
    \end{tabular}
    \vspace{-0.1in}
    \caption{{\small \nf{Five} traces from the DCC example.}}
    \label{fig:traces}
    \vspace{-0.2in}
\end{figure}

\vskip 0.05in
\noindent
\textbf{Syntax.}
Let $\mathbb{I}$ be a set of non-empty intervals over $\mathbb{N}$. An \emph{interval} $I \in \mathbb{I}$ can be expressed as $[b, b')$ where $b \in \mathbb{N}$ and $b' \in \mathbb{N} \cup \infty$.  A \emph{signature} $S$ is a tuple $(C, R, \iota)$, where $C$ is a set of constants  and $R$ is a finite set of predicate symbols 
(for relation), respectively. Without loss of generality, we assume all constants are from the integer domain $\mathbb{Z}$ where the theory of linear integer arithmetic (LIA) holds. The function $\iota : R \rightarrow \mathbb{N}$ associates each predicate symbol $r \in R$ with an arity $\iota(r) \in \mathbb{N}$. Let \nf{$\variables$} be a countable infinite set of variables from domain $\mathbb{Z}$ and a term $t$ is defined inductively as $t: \: c \: | \:  v \: | \: t + t \: | \:  c \times t$. We denote $\bar{t}$ as a vector of terms and $\bar{t}^{\;k}_{x}$ as the vector that contains $x$ at index $k$. 
The syntax of \mfotl formulas is defined as follows:
    \emph{(1)} $\top$ and $\bot$, representing values ``true'' and ``false''; 
    \emph{(2)} $t = t'$ and $t > t'$, for terms $t$ and $t'$;  
    \emph{(3)} $r(t_1 ... t_{\iota(r)})$ for $r \in R$ and terms $t_1 ... t_{\iota(r)}$\label{rule:guarded};  
    \emph{(4)} $\phi \wedge \psi$, $\neg \phi$ for \mfotl formulas $\phi$ and $\psi$; 
    \emph{(5)} $\exists x . (r(\bar{t}^{\;k}_{x}) \wedge \phi)$ for \mfotl formula $\phi$, relation symbol $r \in R$, variable $x \in \variables$ and a vector of terms $\bar{t}^{k}_{x}$ s.t.  $x = \bar{t}^{k}_{x}[k]$; 
    and \emph{(6)} $\until {\phi} {\psi} {I}$ (until), $\since {\phi} {\psi} {I}$ (since), $\opnext{\phi} {I}$ (next), $\prev{\phi} {I}$ (previous) for \mfotl formulas $\phi$ and $\psi$, and an interval $I \in \mathbb{I}$.

We consider a restricted form of quantification (syntax rule~\textit{(5)}, above) similar to guarded quantification~\cite{Halle-12}. Every existentially quantified variable $x$ must be guarded by some relation $r$ (i.e., for some $\bar{t}$, $r(\bar{t})$ holds and $x$ appears in $\bar{t}$). Similarly, universal quantification must be guarded as $\forall x . (r(\bar{t}) \Rightarrow \phi)$ where $x \in \bar{t}$.  Thus, $\nf{\neg} \exists x . \neg r(x)$ (and $\forall x . r(x))$ are not allowed.

The temporal operators $\mathcal{U}_I$, $\mathcal{S}_I$, $\CIRCLE_I$ and $\Circle_I$ require the satisfaction of the formula within the time interval given by $I$. We write $[b, )$ as a shorthand for $[b, \infty)$; if $I$ is omitted, then the interval is assumed to be $[0, \infty)$. Other classical unary temporal operators $\lozenge_I$ (eventually), $\square_I$ (always), and $\blacklozenge_I$ (once) are defined as follows:  $\eventually{\phi}{I} = \until{\top}{\phi}{I}$, $\always{\phi}{I} = \neg \eventually{\neg\phi}{I}$, and  $\once{\phi}{I} = \since{\top}{\phi}{I}$.
Other common logical operator such as $\vee$ (disjunction) and $\forall$ (universal quantification) are expressed through negation of $\wedge$ and $\exists$, respectively.

\begin{example}
\label{expl:synthax}
Suppose a data collection centre (DCC) \textit{collect}s and \textit{access}es personal data information with three requirements: $req_0$ stating that no data is allowed to be accessed before the data ID has been collected for 15 days (360 hours); $req_1$: data can only be updated after having been collected or last updated for more than a week (168 hours); and $req_2$:  data value can only be accessed if the value has been collected or updated within a week (168 hours).
The signature $S_{data}$ for DCC contains three binary relations ($R_{data}$): \textit{Collect}, \textit{Update}, and \textit{Access}, such that \textit{Collect}($d$, $v$), \textit{Update}($d$, $v$) and \textit{Access}($d$, $v$) hold at a given time point if and only if data at id $d$ is collected, updated, and accessed with value $v$ at this time point, respectively.  The \mfotl formulas for $P1$, $req_0$, $req_1$ and $req_2$ are shown in Fig.~\ref{fig:formulas}.
For instance, the formula $req_0$ specifies that if a data value stored at id $d$ is accessed, then some data must have been collected and stored at id $d$ at least 360 hours ago ($\blacklozenge_{[360, )}]$).
\end{example}

\vskip 0.03in
\noindent
\textbf{Semantics.}
 A first-order (FO) structure $D$ over the signature $S = (C, R, \iota)$ is comprised of a non-empty domain $\textit{dom}(D) \neq \emptyset$ and an interpretation for $c^D \in \textit{dom}(D)$ and $r^D \subseteq \textit{dom}(D)^{\iota(r)}$ for each $c \in C$ and $r \in R$. The semantics of \mfotl formulas is defined over a sequence of FO structures $\Bar{D}= (D_0,D_1, \ldots )$ and a sequence of natural numbers representing time $\Bar{\tau} = (\tau_0, \tau_1, \ldots)$, where
    (a) $\Bar{\tau}$ is a monotonically increasing sequence;
    (b) $\textit{dom}(D_i) = \textit{dom}(D_{i+1})$ for all $i \ge 0$ (all $D_i$ have a fixed domain); and
    (c) each constant symbol $c \in C$ has the same interpretation across $\bar{D}$ (i.e., $c^{D_i} = c^{D_{i+1}})$. 
Property (a) ensures that time never decreases as the sequence progresses; 
and (b) ensures that the domain is fixed (referred to as $\textit{dom}(\Bar{D})$)
$\Bar{D}$ is similar to timed words in metric time logic (MTL), but instead of associating a set of propositions with each time point, \mfotl uses a structure $D$ to interpret the symbols in the signature $S$. 
The semantics of \mfotl is defined over a trace of timed first-order structures $\sigma = (\Bar{D}, \Bar{\tau})$, where every structure $D_i \in \Bar{D}$
specifies the set of tuples ($r^{D_i}$) that hold for every relation $r$ at time $\tau_i \in \Bar{\tau}$. 
Let $(\Bar{D}, \Bar{\tau})$ denote  an \mfotl trace.

\begin{example}
\label{expl:signature}
Consider the signature $S_{data}$ in the DCC example. Let $\tau_1 = 0$ and $\tau_2 =361$, and let $D_1$ and $D_2$ be two first-order structures with $r^{D_1} = \textit{Collect}(0, 0)$ and $r^{D_2} = \textit{Access}(0, 0)$, respectively. The trace $\sigma_1 = ((D_1, D_2), (\tau_1, \tau_2))$ is a valid trace shown in Fig.~\ref{fig:traces} and representing two timed relations: (1) data value 0 collected and stored at id 0 at hour 0 and (2) data value 0 is read by accessing id 0 at hour 361.  
\end{example}

 \vspace{-0.05in}
 A \emph{valuation function} $v:\! \nf{\variables} \rightarrow\! \textit{dom}(\bar{D})$  maps a set \nf{$\variables$} of variables  to their interpretations in the domain $\textit{dom}(\bar{D})$.
For vectors $\bar{x}\! =\! (x_1,\! \ldots,\! x_n)$ and $\bar{d}\! =\! (d_1,\! \ldots,\! d_n) \in \textit{dom}(\bar{D})^n$, the \emph{update operation} $v[\bar{x} \rightarrow \bar{d}]$ produces a new valuation function $v'$ s.t. $v'(x_i) = d_i$ for $ 1 \le i \le n$, and $v(x') = v'(x')$ for every $x' \notin \bar{x}$. 
For any  constant $c$,  $v(c) = c^D$.
Let $\bar{D}$ be a sequence of FO structures over signature $S = (C, R, \iota)$ and $\bar{\tau}$  be a sequence of natural numbers. Let $\phi$ be an \mfotl formula over $S$, $v$ be a valuation function and $i \in \mathbb{N}$. A fragment of the relation $(\bar{D}, \bar{\tau}, v, i) \models \phi$ is defined in Fig.~\ref{tab:semantic}.  

The operators $\CIRCLE_I$, $\Circle_I$, $\mathcal{U}_I$ and $\mathcal{S}_I$ are augmented with an interval $I\in \mathbb{I}$ which defines the satisfaction of the formula within a time range specified by $I$ relative to the current time at step $i$, i.e., $\tau_i$.  

\begin{figure}[t]
\begin{tabular}{lcl}
            \centering
              $(\bar{D}, \bar{\tau}, v, i) \models t = t'$ & iff & $v(t) = v(t')$\\
              $(\bar{D}, \bar{\tau}, v, i) \models t > t'$ & iff & $v(t) > v(t')$\\
              $(\bar{D}, \bar{\tau}, v, i) \models r(t_1,..,t_{\iota(r)})$ & iff & $r(v(t_1),..,v(t_{i(r))}) \in r^{D_i}$\\
             $(\bar{D}, \bar{\tau}, v, i) \models \neg \phi$ & iff & $(\bar{D}, \bar{\tau}, v, i) \not\models \phi$ \\
             $ (\bar{D}, \bar{\tau}, v, i) \models \phi \wedge \psi $ & iff & $(\bar{D}, \bar{\tau}, v, i) \models \phi$ and  $(\bar{D}, \bar{\tau}, v, i) \models \psi$\\
             $ (\bar{D}, \bar{\tau}, v, i) \models \exists x \cdot (r(\bar{t}^{\nf{k}}_{x}) \wedge \phi)  $ & iff & $(\bar{D}, \bar{\tau}, v[x \rightarrow d], i) \models (r(\bar{t}^{\nf{k}}_{x})) \wedge \phi$ for some $d \in \textit{dom}({\bar{D}})$\\
             $ (\bar{D}, \bar{\tau}, v, i) \models \opnext{\phi} {I}  $ & iff & $(\bar{D}, \bar{\tau}, v, i+1) \models \phi$ and $\tau_{i+1} - \tau_{i} \in I$\\
             $ (\bar{D}, \bar{\tau}, v, i) \models \prev{\phi} {I}  $ & iff & $i \ge 1$ and $(\bar{D}, \bar{\tau}, v, i-1) \models \phi$ and $\tau_{i} - \tau_{i-1} \in I$\\
             $ (\bar{D}, \bar{\tau}, v, i) \models \until{\phi} {\psi} {I}  $ & iff & exists $j \ge i$ and $(\bar{D}, \bar{\tau}, j , v) \models \psi$ and $\tau_j - \tau_i \in I$\\
            & & and for all $k \in \mathbb{N}$ $i \le k < j \Rightarrow (\bar{D}, \bar{\tau}, k, v) \models \phi$\\
             $ (\bar{D}, \bar{\tau}, v, i) \models \since{\phi} {\psi} {I}  $ & iff & exists $j \le i$ and $(\bar{D}, \bar{\tau}, j , v) \models \psi$ and $\tau_i - \tau_j \in I$\\
            & & and for all $k \in \mathbb{N}$ $i \ge k > j \Rightarrow (\bar{D}, \bar{\tau}, k, v) \models \phi$
\end{tabular}
\caption{\mfotl semantics.}
\label{tab:semantic}
\vspace{-0.15in}
\end{figure}

\begin{definition}[\mfotl Satisfiability]
An \mfotl formula $\phi$ is \emph{satisfiable} if there exists a sequence of FO structures $\Bar{D}$ and natural numbers $\Bar{\tau}$, and a valuation function $v$ such that $(\Bar{D}, \Bar{\tau}, v, 0) \models \phi$. $\phi$ is \emph{unsatisfiable} otherwise. 
\end{definition}

\begin{example}
\label{expl:satisfiability}
In the DCC example, the 
\mfotl formula $req_{0}$ is \emph{satisfiable} because  $(\Bar{D}, \Bar{\tau}, v, 0) \models req_{0}$ (where $\sigma_1 = (\Bar{D}, \Bar{\tau})$ in Fig.~\ref{fig:traces}). 
Let $req_{0}'$ be another \mfotl formula: $\eventually{\exists j . (\textit{Access}(0, j))}{[0, 359]}$.  The formula $req_{0}' \wedge req_{0}$ is \emph{unsatisfiable} because if data stored at id $0$ is accessed between 0 and 359 hours, then it is impossible to collect the data at least 360 hours prior to its access.
\end{example}

\section{Bounded Satisfiability Checking Problem}
\label{sec:bsc}

The satisfiability of MFOTL properties is generally undecidable since MFOTL is expressive enough to describe the blank tape problem~\cite{post-47} (which has been shown to be undecidable). Despite the undecidability result, we can derive a bounded version of the problem, \textit{bounded satisfiability checking} (BSC), for which a sound and complete decision procedure exists. When facing a hard instance for  satisfiability checking, the solution to BSC provides bounded guarantees (i.e., whether a solution exists within a given bound). In this section, we first define satisfiability checking and then the BSC problem for MFOTL formulas. 
\emph{Satisfiability checking}~\cite{DBLP:journals/tosem/PradellaMP13} is a verification technique that extends model checking by
replacing a state transition system with a set of temporal logic formulas.  
In the following, we define  satisfiability checking of \mfotl formulas.
\begin{definition}[Satisfiability Checking of \mfotl Formulas] \label{def:MFOTLMTC}
Let $P$ be an \mfotl formula over a signature $S = (C, R, \iota)$, and let $\reqs$ be a set of \mfotl requirements over $S$. $\reqs$ complies with $P$ (denoted as $\reqs \Rightarrow P) $ iff $\bigwedge_{\psi\in \reqs} \psi \wedge \neg P$ is \emph{unsatisfiable}.  We call a solution to $\bigwedge_{\psi\in \reqs} \psi \wedge \neg P$, if one exists,  a \emph{counterexample} to $\reqs \Rightarrow P$.
\end{definition}
\begin{example}
\label{expl:satisfiabilitychecking}
Consider our DCC system requirements and the privacy data property $P1$ stating that if personal health information is not accurate or not up-to-date, it should not be accessed (see Fig.~\ref{fig:formulas}).
$P1$ is not respected by the set of DCC requirements $\{req_0, req_1, req_2\}$ because $\neg P1 \wedge req_0 \wedge req_1 \wedge req_2$ is \emph{satisfiable}. 
The counterexample $\sigma_2$ (shown in Fig.~\ref{fig:traces}) indicates that data can be re-collected, and the re-collection does not have the same time restriction as the updates. If a fourth policy requirement $req_3$ (Fig.~\ref{fig:formulas}) is added to prohibit re-collection of collected data, then property $P1$ would be respected (i.e., $\{req_0, req_1, req_2, req_3\} \Rightarrow P1$).
\end{example}

\begin{definition}[Finite trace and bounded trace]
Given a trace $\sigma = (\bar{D}, \bar{\tau}, v)$, 
we use $vol(\sigma)$ \nf{(the \emph{volume} of $\sigma$)}, to denote the total number of times that any relation holds
across all FO structures in $\bar{D}$ (i.e., $\sum_{r \in R}\sum_{D_i \in \bar{D}} (|r^{D_i}|)$). The trace $\sigma$ is \emph{finite} if $vol(\sigma)$ is finite. The trace is \emph{bounded by volume $\bound \in \mathbb{N}$} 
if and only if $vol(\sigma) \le \bound$.
\end{definition}


\begin{example}
\label{expl:finitetrace}
The volume of trace $\sigma_3$ in Fig.~\ref{fig:traces}, $vol(\sigma_3) = 3 $ since there are three relations: \textit{Collect}(1, 15), \textit{Update}(1, 0), and \textit{Access}(1, 15). 
Note that the volume 
\nf{is the total number of tuples that hold for any relation across all time points}; \nf{ multiple tuples can thus hold for multiple relations for a single time point}.
\end{example}


\begin{definition}[Bounded satisfiability checking of \mfotl properties]\label{def:BMFOTLMTC}
Let $P$ be an \mfotl property, $\reqs$ be a set of \mfotl  requirements, and $\bound$ be a natural number. The \emph{bounded satisfiability checking problem} determines the existence of a counterexample $\sigma$ to $\reqs \Rightarrow P$ such that
 $vol(\sigma) \le \bound$.  
\end{definition}

\section{Checking Bounded Satisfiability}
\label{sec:approach}

In this section, we present an overview of the bounded satisfiability checking (BSC) process that translates the MFOTL formula into \textit{first-order logic with relational objects} (\tfol) formulas, and looks for a satisfying solution for the \tfol{} formulas.
Then, we provide the translation of MFOTL formulas to \tfol and discuss the process complexity.
 
\subsection{Overview of BSC for MFOTL Formulas}

\tikzset{
    tool/.style={draw,rectangle,rounded corners,minimum width=2cm,minimum height=0.9cm,align=center,text width=1.6cm,fill=black!20,font=\small\sffamily},
    toold/.style={draw=red,line width=2pt, dotted,rectangle,rounded corners,minimum width=2cm,minimum height=0.9cm,align=center,text width=1cm,fill=red!20,font=\small\sffamily},
    tool-n/.style={draw,rectangle,rounded corners,minimum width=2cm,minimum height=0.9cm,align=center,text width=2cm,fill=white,font=\small\sffamily},
    tool_carla/.style={draw,rectangle,rounded corners,minimum width=1.4cm,minimum height=0.9cm,align=center,text width=0.9cm,fill=white,font=\small\sffamily},
    artifact/.style={draw=white,trapezium,trapezium left angle=82,trapezium right angle=98,text width=3.2cm,align=center,font=\small\sffamily},
    artifactd/.style={draw=white,trapezium,trapezium left angle=82,trapezium right angle=98,text width=2.8cm,align=center,font=\small\sffamily},
    artifactmm/.style={draw=white,trapezium,trapezium left angle=82,trapezium right angle=98,text width=2.6cm,align=center,font=\small\sffamily},
    artifactbb/.style={draw=white,trapezium,trapezium left angle=82,trapezium right angle=98,text width=2cm,align=center,font=\small\sffamily},
    artifactbbb/.style={draw=white,trapezium,trapezium left angle=82,trapezium right angle=98,text width=2.3cm,align=center,font=\small\sffamily},
    artifact-b/.style={draw=white,trapezium,trapezium left angle=82,trapezium right angle=98,text width=2.8cm,align=center,font=\small\sffamily},
    artifact-s/.style={draw=white,trapezium,trapezium left angle=82,trapezium right angle=98,text width=1.6cm,align=center,font=\small\sffamily},
    progl/.style={draw,tape,tape bend top=none,align=center,text width=2cm,fill=gray!20,font=\small\sffamily}
}


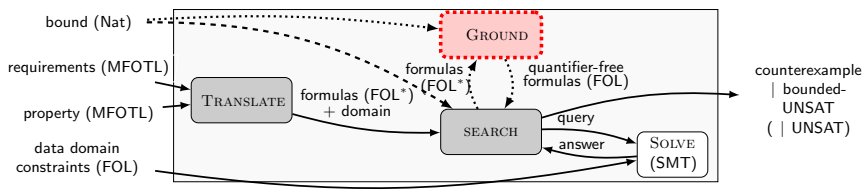
\begin{figure*}[t]
    \centering
\scalebox{.67}{
\begin{tikzpicture}[x=2.25cm,y=0.9cm]
\draw[fill=lightgray!10] (0.6,5.7) rectangle (5.4,1.9);


\node[artifact] (req) at (-0.15,4.4) {requirements (MFOTL)};
\node[artifactmm, below = 10 pt of req] (prop) {property (MFOTL)};
\node[artifactd] (datacons) at (-0.25,2.4) {data domain constraints (FOL)};
\node[artifactbb, above = 10 pt of req] (bound) {bound (Nat)};
\node[tool] (trans) at (1.2,3.7) {\textsc{Translate}};
\node[tool] (search) at (3.4,3) {\textsc{search}};
\node[toold, above = 28 pt of search] (ground)  {\textsc{Ground}};
\node[tool_carla] (smt) at (5,2.5) {\textsc{Solve} \\ (SMT)};
\node[artifactbbb] (out) at (6.2,3.7) {counterexample $\; | \; $bounded-UNSAT \nf{($\; | \; $UNSAT)}};

\node[rotate=0,font=\sffamily] at (2.9,4.4) {{\small formulas}};
\node[rotate=0,font=\sffamily] at (3,4) {{\small (\tfol)}};
\node[rotate=0,font=\sffamily] at (2.22,3.8) {{\small formulas (\tfol)}};
\node[rotate=0,font=\sffamily] at (2.22,3.45) {{\small $+$ domain}};
\node[rotate=0,font=\sffamily] at (4.15,4.4) {{\small quantifier-free}}; 
\node[rotate=0,font=\sffamily] at (4.15,4.1) {{\small formulas (FOL)}};;
\node[rotate=0,font=\sffamily] at (4.15,3.26) {{\small query}};
\node[rotate=0,font=\sffamily] at (4.2,2.7) {{\small answer}};



\begin{scope}[every path/.style={-latex}]
\draw [line width=1pt] (trans) edge [bend right=8](search)
       (req) edge (trans)
       (prop) edge (trans)
       (datacons) edge [bend right=9](smt)
       (search) edge [bend left=10](out)
       (search) edge [bend left=10](smt)
       (smt) edge  [bend left=10](search)
      ;
      
\draw[dotted,draw=red,line width=1.25pt]
       (bound) edge [bend left=5] (ground)
       (search) edge [bend left] (ground)
       (ground) edge [bend left] (search)
       ;
       
\draw[dashed,draw=mymauve,line width=1.25pt]
       (bound) edge [bend left=16] (search);
      
      
\end{scope}
\end{tikzpicture}
}
\vspace{-0.1in}

\caption{\small Overview of the naive and our incremental ($\mainAlg$) MFOTL bounded satisfiability checking approaches. Solid boxes and arrows are shared between the two approaches. Blue dashed arrow is specific to the naive approach. Red dotted arrows and the additional red output in bracket are specific to $\mainAlg$.}
\label{fig:architecture-c}
\vspace{-0.15in}
\end{figure*}
We aim to address the bounded satisfiability checking problem (Def.~\ref{def:BMFOTLMTC}),  looking for a satisfying run $\sigma$ within a given volume bound $\bound$ that limits the number of relations in $\sigma$.  
First, we $\textsc{translate}$ the \mfotl formulas to \tfol formulas.
The considered constraints in the formulas include 
those of the system requirements and the legal property, and \emph{optional} data constraints specifying the data \nf{value constraint for a datatype.}
\nf{The data constraints can be defined as a range, a ``small'' data set, or the union/intersection of other data constraints. If data constraints are not specified, then the data value comes from the domain $\mathbb{Z}$. Note that the optional data constraints do not affect the complexity of BSC, but they do help prune unrealistic counterexamples.}
Second, we $\textsc{search}$ for a satisfying solution to the \tfol formula; an SMT solver is used here to determine the satisfiability of the \tfol constraints and the data domain constraints. 
The answer from the SMT solver is analyzed to return an answer to the satisfiability checking problem (a counterexample $\sigma$, or "bounded-UNSAT").

\subsection{Translation of MFOTL to First-Order Logic}
\label{sec:translate}
In this section, we describe the translation target \tfol, the translation rules and prove their correctness.
\vspace*{-1em}
\subsubsection{FOL with Relational Object (FOL*)}
We start by introducing the syntax of \tfol. 
A \emph{signature} $S$ is a tuple $(C, R, \iota)$, where $C$ is a set of constants, $R$ is a set of relation symbols, and $\iota : R \rightarrow \mathbb{N} $ is a function that maps a relation to its arity.   We assume that the domain of constant $C$ is $\mathbb{Z}$, which matches the one for \mfotl, where the theory of linear integer arithmetic (LIA) holds. Let \nf{$\variables$} be a set of variables in the domain $\mathbb{Z}$. A \emph{relational object} $o$ of class $r \in R$ (denoted as $o:r$) is an object with $\iota(r)$ regular attributes and two special attributes, where every attribute is a variable. We assume that all regular attributes are ordered and denote $o[i]$ to be the $i$th attribute of $o$. 
Some attributes are named, and $o.x$ refers to $o$'s attribute with the name `$x$'.  
Each relational object $o$ has two special attributes $o.ext$ and $o.time$. The former is a boolean variable indicating whether $o$ exists in a solution, and the latter is a variable representing the occurrence time of $o$. For convenience, we define a function $\class$($o$) to return the relational object's class.  Let \emph{a} \tfol{} \emph{term} $t$ be defined inductively as $t: \; c \; | \;  v \; | \; o[\nf{k}] \; | \; o.x \; | \; t + t \; | \;  c \times t$ for any constant $c \in C$, any variable $v \in \nf{\variables}$, any relational object $o:r$, any index $\nf{k} \in [1, \iota(r)]$ and any valid attribute name $x$. Given a signature $S$, the syntax of the \tfol{} formulas is defined as follows:
    \emph{(1)} $\top$ and $\bot$, representing values ``true'' and ``false''; 
    \emph{(2)} $t = t'$ and $t > t'$, for term $t$ and $t'$;  
    \emph{(3)} $\phi_f \wedge \psi_f$, $\neg \phi_f$ for \tfol{} formulas $\phi_f$ and ${\psi_f}$; 
    \emph{(4)} $\exists o:r \cdot \: (\phi_f)$ for an \tfol{} formula $\phi_f$  and a class $r$;
    \emph{(5)} $\forall o:r \cdot \: (\phi_f)$ for an \tfol{} formula $\phi_f$  and a class $r$. The quantifiers for \tfol{} formulas are limited to relational objects, as shown by rules (4) \& (5).  
    Operators $\vee$ and $\forall$ can be defined in \tfol as follows: $\phi_f \vee \psi_f = \neg (\neg \phi_f \wedge \neg \psi_f)$ and $\forall o:r \cdot \phi_f = \exists o:r \cdot \neg \phi_f$. We say an \tfol formula is in a \textit{negation normal form} (NNF) if negations ($\neg$) do not appear in front of $\neg$, $\wedge$, $\vee$, $\exists$ and $\forall$. For the rest of the paper, we assume that every \tfol $\phi$ is in NNF.     
    
Given a signature $S$, \emph{a domain} $\domain$ is a finite set of relational objects. An \tfol{} formula \emph{grounded} in the domain $\domain$ (denoted by $\phi_{\domain}$) is a quantifier-free FOL formula that eliminates quantifiers on relational objects using the following rules: (1) $\exists o: r \cdot \: (\phi_f)$  to $ \bigvee_{{o':r} \in \domain} (o'.ext \wedge \phi_f [o \leftarrow o'])$ and (2) $\forall o:r \cdot \: (\phi_f)$ to $\bigwedge_{{o':r} \in \domain} (o'.ext \Rightarrow \phi_f [o \leftarrow o'])$. An \tfol{} formula $\phi_f$ is \emph{satisfiable in $\domain$} if there exists a variable assignment $v$  that evaluates $\phi_{\domain}$ to $\top$ according to the standard semantics of FOL. An \tfol{} formula $\phi_f$ is \emph{satisfiable} if there exists a finite domain $\domain$ such that $\phi_f$ is satisfiable in $\domain$. 
We call $\sigma = (\domain, v)$ \emph{a satisfying solution to $\phi_f$}, denoted as $\sigma \models \phi_f$. Given a solution $\sigma = (D, v)$, we say a relational object $o$ is in $\sigma$, denoted as $o \in \sigma$, if $o \in D$ and $v(o.ext)$ is true.
 The \emph{volume of the solution}, denoted as $vol(\sigma)$, is $|\{o \mid o \in \sigma\}|$.
\begin{example} Let $a$ be a relational object of class $\textit{A}$ with attribute name $val$. The formula $\forall a: A.\: (\exists a':A \cdot \: (a.val < a'.val) \wedge \exists a:A \cdot \: a.val = 0) $ has no satisfying solutions in any finite domain. On the other hand, the formula $\forall a:A \cdot \: (\exists a', a'':A  \cdot \: (a.val = a'.val + a''.val) \wedge \exists a:A \cdot \: a.val = 5) $ has a solution $\sigma = (\domain, v)$ of volume 2, with the domain $\domain = (a_1, a_2)$ and the value function $v(a_1.val) = 5$, $v(a_2.val) = 0$ because if $a \gets a_1$ then the formula is satisfied by assigning $a' \gets a_1, \;a'' \gets a_2$; and  if $a \gets a_2$, then the formula is satisfied by assigning $a' \gets a_2, \;a'' \gets a_2$.
\end{example}
\subsubsection{\textbf{From \mfotl Formulas to \tfol{} Formulas.}}
We now discuss the translation rule from the \mfotl formulas to \tfol formulas. Recall that \mfotl semantics is defined for a time point $i$ on a trace $\sigma = (\bar{D}, \bar{\tau}, v, i)$, where $\bar{D} = (D_1, D_2, \ldots)$ is a sequence of FO structures and $\bar{\tau} = (\tau_1,\tau_2,\ldots)$ is a sequence of time values. The time value of the time point $i$ is given by $\tau_i$, and if $i$ is not specified, then $i = 1$. The semantics of the \tfol formulas is defined for a domain $\domain$ where the information of time is associated with relational objects in the domain. Therefore, the time point $i$ (and its time value $\tau_i$) should be considered during the translation from \mfotl to \tfol since the same \mfotl formula at different time points represents different constraints on the trace $\sigma$. Formally, our translation function $\textsc{translate}$, abbreviated as $T$, translates an \mfotl formula $\phi$ into a function $f: \tau \rightarrow \phi_{f}$, where $\tau \in \mathbb{N}$ and $\phi_{f}$ is an \tfol formula. The translation rules are stated in Fig.~\ref{tab:translate}.

\noindent
\underline{Representing time points in \tfol.} Since \tfol quantifiers are limited to relational objects, to quantify over time points (which is necessary to capture the semantics of \mfotl temporal operators such as $\mathcal{U}$), 
the translated \tfol formulas use a special \emph{internal} class of relational objects $\tp$ (e.g., $\exists o:\tp$). Relational objects of class $\tp$ capture all possible time points in a trace, and they have two attributes, $ext$ and $time$, to record the existence and the value of the time point, respectively. To ensure that every time value in a solution is represented by some relational object of $\tp$, we introduce the \textit{time coverage} \tfol axiom.
\begin{named}[Time coverage] \label{axiom:tc}
Let $\phi_f$ be an \tfol formula and let $\sigma$ be its solution. For every relational object $o \in \sigma$, there exists an object $o'$ of class $\tp$ s.t. $o$ and $o'$ share the same time value. Formally, $\forall o \cdot (\exists o': \tp \cdot o.time = o'.time)$. 
\end{named}
The translation of $\opnext{\phi}{I}$ uses function $\textsc{Next}(t_1, t_2)$ to assert that $t_1$ is the next time value of $t_2$. Formally, $\textsc{Next}(t_1, t_2) = \forall o: \tp \cdot o.time > t_2 \Rightarrow t_1 \le o.time$.  Function $\textsc{Prev}(t_1, t_2)$ for translation of $\prev{\phi}{I}$ is defined similarly. 

\begin{figure}[t]
\scalebox{0.85}{
    \begin{tabular}{lcl}
                \centering
                $T(t = t', \tau_i)$ & $\rightarrow$ & $t = t'$\\
                $T(t > t', \tau_i)$ & $\rightarrow$ & $t > t'$\\        
                $T(r(t_1,..,t_{\iota(r)}), \tau_i)$ & $\rightarrow$ & $\exists o:r \cdot \bigwedge_{j=1}^{\iota(r)}(o.j = t_j) $
                $\wedge (\tau_i = o.time)$ \\
                $ T(\neg\phi, \tau_i)$ & $\rightarrow$ & $\neg T(\phi, \tau_i)$\\
                $T(\phi \wedge \psi, \tau_i)$ & $\rightarrow$ & $T(\phi, \tau_i)\wedge T(\psi, \tau_i)$\\
                $T(\exists x \cdot r(\bar{t}^{\nf{k}}_{x}) \wedge \phi, \tau_i)$ & $\rightarrow$ & $\exists o : r \cdot T( (r(\bar{t}^{\nf{k}}_{x}) \wedge \phi)[x \rightarrow o[\nf{k}]], \tau_i)$\\
                $T(\opnext{\phi}{I}, \tau_i) $ & $\rightarrow$ & $\exists o:\tp \cdot \textsc{Next}(o.time,  \tau_i) \wedge T(\phi, o.time) \wedge (o.time - \tau_{i}) \in I$ \\
                $T(\prev{\phi}{I}, \tau_i) $ & $\rightarrow$ & $\exists o:\tp \cdot  \textsc{Prev}(o.time,  \tau_i) \wedge T(\phi, o.time) \wedge (\tau_{i} - o.time) \in I$\\
                 $T(\until{\phi}{\psi}{I}, \tau_i) $ & $\rightarrow$ & $\exists o: \tp  \cdot (o.time \ge \tau_i \wedge (o.time - \tau_i) \in I \wedge T(\psi, o.time)$\\
                 && and $\forall o': \tp \cdot o'.time \cdot (\tau_i \le o'.time < o.time \Rightarrow T(\phi, o'.time)))$\\
                 $T(\since{\phi}{\psi}{I}, \tau_i) $ & $\rightarrow$ & $\exists o:\tp  \cdot (o.time \le \tau_i \wedge (\tau_i - o.time) \in I \wedge T(\psi, o.time)$\\
                 && and $\forall o': \tp \cdot (\tau_i \ge o'.time > o.time \Rightarrow T(\phi, o'.time)))$\\
                $T(\phi)$ & $\rightarrow$ & $T(\phi, \tau_1)$ 
                \end{tabular}
 }           
\caption{{\small Translation rules from \mfotl to \tfol. 
$\tp$ is an internal class of relational objects used to represent time values at different time points. The predicate $\textsc{Next}(t_1, t_2)$ ($\textsc{Prev}(t_1, t_2)$) asserts that $t_1$ is the next (previous) time value of $t_2$.
}} 
\vspace{-0.2in}
\label{tab:translate}
\end{figure}

\begin{definition}[Mapping from \mfotl trace to \tfol trace]\label{def:solutionmap}
Let an \mfotl trace $(\Bar{D}, \Bar{\tau})$ and a valuation function $v$ be given. A function $M ((\Bar{D}, \Bar{\tau}), v) \rightarrow (\domain, v')$ is \emph{a mapping between an \mfotl trace and an \tfol trace} if 
$M$ satisfies
the following rules:
(1) for every $\tau_i \in \Bar{\tau}$, there exists a relational object $o: \tp \in  \domain$ such that $\tau_i = v'(o.time)$; (2) for every structure $D_{i} \in \Bar{D}$, if a tuple $\bar{t}$ holds for a relation $r$, (i.e., $\bar{t} \in r^{D_i}$), then there exists a relational object $o:r$ such that for $j \in \iota(r)$, $\bar{t}[j] = v'(o[j])$ and $v'(o.time) = \tau_i \wedge v'(o.ext) = \top$; (3)  for every term $t$ defined \nf{for} $v$, $v(t) = v'(T(t, \tau_i))$.
\end{definition}
\noindent The inverse of $M$, denoted as $M^{-1}$, is defined as follows: (1) $\Bar{\tau} = \textsc{sort}(\{v'(o.time) \mid o : TP \in D \cdot v'(o.ext) \})$ and (2) for every relational object $o:r$, if $v'(o.ext)$, then $(v'(o[1]) \ldots v'(o[\iota(r)])) \in r^{D_{i}}$, where $i$ is the index of \nf{the time value} $v'(o.time)$ in $\Bar{\tau}$.
\begin{lemma}\label{lemma: tcorrectness}
Given an \mfotl formula $\phi$, an \mfotl trace $(\Bar{D}, \Bar{\tau})$, a valuation function $v$, and a time point $i$,  the relation $(\bar{D}, \bar{\tau}, v, i) \models \phi$ holds iff there exists a satisfying trace $\sigma = (D, v')$ for the formula $T(\phi, \tau_i)$.
\end{lemma}

\vspace{-0.15in}

\begin{sop}
In the proof, we use $M$ and $M^-1$ (see Def.~\ref{def:solutionmap}) 
to transform an \mfotl solution into an \tfol trace, and show that it is a solution to the translated \tfol formula (and vice versa).  

$\implies$: if $(\bar{D}, \bar{\tau}, v, i) \models \phi$, then it is sufficient to show $(D, v') \gets M(\Bar{D}, \Bar{\tau}, v)$ is an \tfol{} solution. To prove $(D, v')$ is the solution to $T(\phi, \tau_i)$, we consider all the translation rules in Fig.~\ref{tab:translate}. The translated \tfol{} matches the semantics (Fig.~\ref{tab:semantic}) of \mfotl except for the translation of temporal operators (e.g., T($\opnext{\phi}{I}, \tau_i)$ and $T(\until{\phi}{\psi}{I}, \tau_i)$) where instead of quantifying over time points (e.g., $\exists j$ and $\forall k$), internal relational objects of class $\tp$ ($o, o':\tp$) are quantified over. By rule (1) of Dec.~\ref{def:solutionmap}, 
every time point and its time value are mapped to some relational object of class $\tp$. Therefore, the quantifiers on time points can be translated into the quantifiers on the relational objects of $\tp$. The mapped solution $(D, v')$ also satisfies Axiom~\ref{axiom:tc} because if a tuple $\bar{t}$ holds for some relation $r$ at some time $\tau$ in the \mfotl trace $(\bar{D}, \bar{\tau})$, then there exists a time point $i \in [1, |\bar{\tau}|]$ such that $\tau_i = \tau$. Therefore, by rule (1) of $M$, $\tau_i$ is represented by some $o :\tp$.

$\Longleftarrow$:  if $(D, v') \models T(\phi, \tau_i)$, then it is sufficient to show that the \mfotl trace $(\Bar{D}, \Bar{\tau}, v) \gets M^{-1}(D, v')$ satisfies $\phi$ at point $i$ (i.e., $(\bar{D}, \bar{\tau}, v, i) \models \phi$).  To prove $(\Bar{D}, \Bar{\tau}, v, i) \models \phi$, we consider all the translation rules in Fig.~\ref{tab:translate}. The translated \tfol{} formula matches the semantics of \mfotl (Fig.~\ref{tab:semantic}) except for the difference between the time points and the relational objects of class $\tp$. By Axiom~\ref{axiom:tc}, every relational object's time is captured by some time point, and by rule (2) of $M^-1$, every relational object is mapped onto some structure $D_i$ at some time $\tau_i$ by $M$. Therefore, $(\Bar{D}, \Bar{\tau}, v, i) \models \phi$.
\end{sop}
\begin{theorem}[Translation Correctness]\label{Thm:tcorrectness}
Given an \mfotl formula $\phi$ and an \mfotl trace $\sigma$, let $M(\sigma)$ be the \tfol solution mapped from $\sigma$ using function $M$ (Def.~\ref{def:solutionmap}). Then (1) $\sigma \models \phi$ if and only if $M(\sigma) \models T(\phi)$, and (2)  $vol(\sigma) = vol(M(\sigma)) - |\{ o : \tp \in M(\sigma) \}|$,  
where $|\{ o : \tp \in M(\sigma) \}|$ is the number of relational objects of the internal class $\tp$ in the solution $M(\sigma)$. 
\end{theorem}
\noindent
\textit{Proof.} Statement (1) of Thm.~\ref{Thm:tcorrectness} is a direct consequence of Lemma~\ref{lemma: tcorrectness}. Statement (2) is the result of rule (2) in Def.~\ref{def:solutionmap} because every relational object in the \tfol solution, except for the internal ones, i.e., $o: \tp$, has a one-to-one correspondence to tuples that hold for some relation in the \mfotl solution. $\qed$

For the rest of the paper, we assume that the internal relational objects of class $\tp$ do not count toward the volume of the \tfol{}, i.e., $vol(\sigma) = vol(T(\sigma))$.

\begin{example}
\label{expl:translate}
Consider a formula 
$exp = \always{\forall d \cdot (A(d) \implies \eventually{B(d)}{[5,10]})}{}$,
where $A$ and $B$ are unary relations.  The translated \tfol formula $T(exp)$ is:
    $\forall o: \tp \cdot \forall a : A \cdot  (o.time = a.time  \Rightarrow \exists o':\tp \cdot b: B \cdot o'.time = b.time \wedge a[1] = b[1] \wedge \; o.time+5 \le o'.time \le o.time+10)$. Since $o.time = a.time$ and $o'.time = b.time$, we can substitute $o.time$ and $o'.time$ with $a.time$ and $b.time$ in $T(exp)$, respectively. Then, the formula contains no reference to $o$ and $o'$, and we can safely drop the quantified $o$ and $o'$ (we can drop existential quantified $\tp$ relational object because of the time coverage axiom). The simplified formula is: $ \forall a : A \cdot \exists  b: B \cdot a[1] = b[1] \wedge \; a.time+5 \le b.time \le a.time+10$. 

This is important for designing system requirements that comply with LPs.
\end{example}

Given an \mfotl{} property $P$ and a set $Reqs$ of \mfotl requirements, and a volume bound $\bound$, the BSC problem can be solved by searching for a satisfying solution $v'$ for the \tfol{} formula $T(\neg P) \bigwedge_{\psi \in Reqs} T(\psi)$ in a domain $\domain$ with at most $\bound$ relational objects.

\vspace*{-1em}
\subsection{Checking \mfotl Satisfiability: A Naive Approach}
\label{sec:complexity}
Below, we define a naive procedure $\naiveAlg$ (shown in Fig.~\ref{fig:architecture-c}) for checking satisfiability of \mfotl formulas translated into \tfol. We then discuss the complexity of this naive procedure. Even though we do not use $\naiveAlg$ in this paper, its complexity constitutes an upper bound for our approach proposed in Sec.~\ref{sec:incremental}.

\vskip 0.05in
\noindent
{\bf Searching for a satisfying solution.} 
Let $\phi_f$ be an \tfol{} formula translated from an \mfotl{} formula $\phi$, and let $\bound$ be the volume bound.
$\naiveAlg$ solves $\phi_f$ via quantifier elimination. 
The number of relational objects in any satisfying solution of $\phi_f$ should be at most \nf{$\bound$}. Therefore, $\naiveAlg$ grounds the \tfol formulas within a domain of $\bound$ relational objects (see Sec.~\ref{sec:translate}), and then uses an SMT solver to check satisfiability of the grounded formula. If the domain has multiple classes of relational objects, we can unify them by introducing a ``superposition'' class whose attributes are the union of the attributes of all classes and a special ``name'' attribute to indicate the class represented by the superposition.

\vskip 0.05in
\noindent
{\bf Complexity.} 
The size of the quantifier-free formula is $O(\bound^k)$, where $k$ is the maximum depth of quantifier nesting. 
Since the background theory used in $\phi$ is restricted to linear integer arithmetic, solving the formula is NP-hard~\cite{DBLP:journals/jacm/Papadimitriou81}. 
Because $T$ (Tab.~\ref{tab:translate}) is linear in the size of the formula $\phi$, $\naiveAlg$ is NP-complete w.r.t. the size of the grounded formula, $\bound^k$.

\section{Incremental Search for Bounded Counterexamples}
\label{sec:incremental}
The naive BSC approach ($\naiveAlg$) proposed in Sec.~\ref{sec:complexity} is inefficient for solving the translated \tfol formulas given a large bound $n$ due to the size of the ground formula. 
Moreover, $\naiveAlg$ cannot detect unbounded unsatisfiability, and cannot provide optimality guarantees on the volume of counterexamples which are important for establishing the proof of unbounded correctness and localizing faults~\cite{Gastin-Moro-Zeitoune-04}, respectively. 
In this section, we propose an incremental procedure $\mainAlg$, which can detect unbounded unsatisfiability and provide the shortest counterexamples. An overview of $\mainAlg$ is given in Fig.~\ref{fig:architecture-c}.

$\mainAlg$ maintains an under-approximation of the search domain and the \tfol constraints.  
It uses the search domain to ground the \tfol constraints, and an SMT solver to determine the satisfiability of the grounded constraints. 
It analyzes the SMT result and accordingly either expands the search domain, refines the \tfol constraints, or returns an answer to the satisfiability checking problem (a counterexample $\sigma$, ``bounded-UNSAT", or ``UNSAT").
The procedure continues until an answer is obtained ($\sigma$ or UNSAT), or until the domain exceeds the bound $\bound$, in which case a ``bounded-UNSAT'' answer is returned. 

In the following, we describe $\mainAlg$ in more detail.
We explain the key component of $\mainAlg$, computing over- and under-approximation queries, in Sec.~\ref{subsec:OverandUnder}.  We discuss the algorithm itself in 
Sec.~\ref{subsec:main} and illustrate it in
Sec.~\ref{subsec:example}.
We prove its soundness (Thm.~\ref{thm:soundness}), completeness (Thm.~\ref{thm:termination}), and solution optimality (Thm.~\ref{Thm:optimality}) in Sec.~\ref{ap:SoundTerminateOpt}.  


\vspace{-0.15in}
\subsection{Over- and Under-Approximation} \label{subsec:OverandUnder}
\vspace{-0.05in}
 $\naiveAlg$ grounds the input \tfol formulas in a fixed domain $\domain$ (fixed by the bound $\bound$).  Instead,
 $\mainAlg$ under-approximates $\domain$ to $\domainUnder$ such that $\domainUnder \subseteq \domain$.
 With $\domainUnder$, we can create an over- and an under-approximation query to the bounded satisfiability checking problem.
 Such queries are used to check the satisfiability of \tfol formulas with domain $\domainUnder$. $\mainAlg$ starts with a small domain $\domainUnder$ and gradually expands it until either SAT or UNSAT is returned, or the
 domain size exceeds some limit (bounded-UNSAT). 
 
\vskip 0.05in
\noindent
\textbf{Over-approximation.}
Let $\groundinput$ be an \tfol formula, and $\domainUnder$ be a domain of relation objects. The procedure \textsc{Ground}, $\groundAlg$($\groundinput$, $\domainUnder$), encodes $\groundinput$ into a quantifier-free FOL formula $\phi_g$ s.t. the unsatisfiability of $\phi_g$ implies the unsatisfiability of $\groundinput$. We call $\phi_g$ an \emph{over-approximation} of $\groundinput$. 
The procedure $\groundAlg$ (Alg.~\ref{alg:ground}) recursively traverses the syntax tree of 
the input \tfol formula from top to bottom.
 
To eliminate the existential quantifier in $\exists o: r \cdot \phi'_f$  (L:\ref{ln:extdiscovery}), $\groundAlg$ creates a new relational object $o'$ of class $r$ (L:~\ref{ln:extnewClass}), and replaces $o$ with $o'$ in $\phi'_f$ (L:\ref{ln:extreplacement}). 
To eliminate the universal quantifier in $\forall o: r \cdot \phi'_f$ (L:~\ref{ln:unidiscovery}), $\groundAlg$ grounds the formula in $\domainUnder$. More specifically, $\groundAlg$ expands the quantifier into a conjunction of clauses where each clause is $o'.ext \Rightarrow \phi'_f[o \leftarrow o']$ (i.e., $o$ is replaced by $o'$ in $\phi'_f$) 
 for each relational object $o'$ of class $r$ in $\domainUnder$ (L:~\ref{ln:forall}).  Intuitively, an existentially quantified relational object is instantiated with a new relational object, and a universally quantified relational object is instantiated with every existing relational object of the same class in $\domainUnder$, which does not include the ones instantiated during $\groundAlg$. 


\begin{lemma} [Over-approximation Query] \label{lemma:overapprox}
For an \tfol formula $\groundinput$, and a domain $\domainUnder$, if $\phi_g = \groundAlg (\groundinput, \domainUnder)$ is UNSAT, then so is $\groundinput$. 
\end{lemma}
  
\noindent
\textbf{Under-approximation.}
Let $\groundinput$ be an \tfol formula, and $\domainUnder$ be a domain. The over-approximation $\phi_g = \groundAlg(\groundinput, \domainUnder)$ contains a set of new relational objects introduced by $\groundAlg$ (L:\ref{ln:extnewClass}), denoted by $NewRs$. Let \textsc{NoNewR}($NewRs$, $\domainUnder$) be constraints that enforce that every new relational object $o_1$ in $NewRs$ be semantically equivalent to some relational objects $o_2$ in $\domainUnder$.
Formally: the predicate $\textsc{NoNewR}(NewRs, \domainUnder)$ is defined as $\bigwedge_{o_1 \in NewRs} \bigvee_{o_2 \in \domainUnder} (o_1 \equiv o_2)$, where the semantically equivalent relation between $o_1$ and $o_1$  (i.e., $o_1 \equiv o_2$) is defined as 
$\class(o_1) = \class(o_2)$ and $\bigwedge_{i=1}^{\iota(\class(o))}(o_1[i] = o_2[i]) \wedge o_1.ext = o_2.ext \wedge o_1.time = o_2.time$ (where the $\class (o)$  returns the class of $o$). 
Let  $\phi_g^{\bot} = \phi_g \wedge \textsc{NoNewR}(NewRs, \domainUnder).$ If $\phi_g^{\bot}$ has a satisfying solution, then there is a solution for $\groundinput$. We call $\phi_g^{\bot}$ an \emph{under-approximation} of $\groundinput$ and denote the procedure for computing it by $\textsc{UnderApprox}(\groundinput, \domainUnder)$.
\begin{lemma} [Under-Approximation Query] \label{lemma:underapprox}
For an \tfol formula $\groundinput$, and a domain $\domainUnder$, let $\phi_g = \groundAlg (\groundinput, \domainUnder)$ and $\phi_g^{\bot} = \underapproxFunc(\groundinput, \domainUnder)$. If $\sigma$ is a solution to $\phi_g^{\bot}$, then there exists a solution to $\groundinput$.
\end{lemma}
The proofs of Lemma~\ref{lemma:overapprox} and ~\ref{lemma:underapprox} are in Sec.~\ref{appendix:over}

\noindent
Suppose, for some domain $\domainUnder$, that an over-approximation query $\phi_g$ for an \tfol formula $\groundinput$ is satisfiable while the under-approximation query $\phi_g^{\bot}$ is UNSAT. 
Then, the solution to $\phi_g$ provides hints on how to expand $\domainUnder$ to potentially obtain a satisfying solution for $\groundinput$, as captured in Cor.~\ref{cor:DomainExpansion}.

\begin{corollary}[Necessary relational objects] \label{cor:DomainExpansion}
For an \tfol formula $\groundinput$ and a domain $\domainUnder$, let $\phi_g$ and $\phi_g^{\bot}$ be the over- and under-approximation queries of $\groundinput$ based on $\domainUnder$, respectively. Suppose $\phi_g$ is satisfiable and $\phi_g^{\bot}$ is UNSAT, then every solution to $\groundinput$ contains some relational object in formula $\phi_g$ but not in $\domainUnder$.
\end{corollary}


\begin{algorithm}[t]
            	\caption{\small $\mainAlg$: search for a bounded (by $\bound$) solution to $T(\negp)\bigwedge_{\psi \in \reqs} T(\psi)$.}
            	\small
            	 \hspace*{\algorithmicindent} {\scriptsize \textbf{Input} an MFOTL formula $\neg P$, and MFOTL requirements $\reqs = \{\psi_1, \psi_2, ...\}$ .\hfill\mbox{}}\\
            	  \hspace*{\algorithmicindent} {\scriptsize \textbf{Optional Input} $\bound$, the volume bound, and  data constraints $\dataDomainConstraint$.\hfill\mbox{}}\\
            	  \hspace*{\algorithmicindent} {\scriptsize \textbf{Output} a counterexample $\sigma$, UNSAT or bounded-UNSAT.\hfill\mbox{}}\\
            	  \vspace{-0.3in}
            	  \begin{multicols}{2}
            	\begin{algorithmic}[1]
            	    \State $\reqs_f \gets \{\ \psi_f = \translate(\psi) \; | \; \psi \in \reqs\} $ 
            	    \State $\negpf \gets \translate (\negp)$
            	    \State $\reqsUnder \gets \emptyset$ \label{ln:initReq} {\color{red}\textit{//initially empty requirement}}
            		\State $\domainUnder \gets \emptyset$ \label{ln:initDomain} {\color{red}\textit{//initially empty domain}}
            		\While {$\top$}  \label{ln:loopStart}
            		\State $\underinput \gets \negpf \wedge \reqsUnder$ 
            		\State $\phi_{g} \gets \groundAlg(\underinput, \domainUnder)$ {\color{red}\textit{//over-approx.}}\label{ln:computeOverground} 
            		\State $\phi_{g}^{\bot} \gets\underapproxFunc(\underinput, \domainUnder)$  {\color{red}\textit{//under-approx.}}\label{ln:computeUnderground}
            		\If {\textsc{Solve}($\phi_{g} \wedge \dataDomainConstraint$) = UNSAT}   \label{ln:overUNSAT}
            		\State \Return UNSAT \label{ln:returnUNSAT}
            		\EndIf \label{ln:domainCheckFinish}
            		\State $\sigma \gets \textsc{Solve}(\phi_{g}^{\bot} \wedge \dataDomainConstraint)$ \label{ln:solve2}
            		\If {$\sigma$ = UNSAT} {\color{red}\textit{//expand $\domainUnder$}}  \label{ln:expansionstart}         
            		        \State $\sigma_{min} \gets \textsc{Minimize}(\phi_g)$ \label{ln:minisolution}
            		        \State {\color{red} \textit{//expand based on $\sigma_{min}$}}
            		        \State $\domainUnder$ +=  $\{o \; | \; o \in \sigma_{min}\}$ \label{ln:expandingdomain}
            		        \If{$vol(\sigma_{min}) > \bound$ }   \label{ln:exceedsLimit}   
            		        \State \Return bounded-UNSAT \label{ln:returnboundedUNSAT}
            		        \EndIf 
            		\Else \label{ln:expansionfinish} {\color{red}\textit{//check all requirements}}
            		    \If {$\sigma \models \psi_f$ for $\psi_f \in \reqs_f$} \label{ln:ReqCheck}
            		    \State \Return $\sigma$ \label{ln:returnSAT}
            		    \Else \label{ln:underSAT}
            		    \State $lesson \gets \psi_f$ for some $\sigma \not\models \psi_f$  \label{ln:findviolatingrule}  
            		    \State $\reqsUnder$.add($lesson$) \label{ln:ruleRefinement}
            		    \EndIf \label{ln:ReqCheckFinish}
            		\EndIf 
              \EndWhile \label{ln:loopEnd}
            	\end{algorithmic} 
            	\end{multicols}
            	\label{alg:main}
            \end{algorithm} 


\begin{algorithm}[t]
                	\caption{\small $\groundAlg$: ground a NNF \tfol  formula $\groundinput$ in a domain $\domainUnder$.}
                	  \small
                	  \hspace*{\algorithmicindent} {\scriptsize \textbf{Input} an \tfol formula $\groundinput$ in NNF, and a domain of relational objects $\domainUnder$ .\hfill\mbox{}}\\
                	  \hspace*{\algorithmicindent} {\scriptsize \textbf{Output} a grounded quantifier-free formula  $\phi_g$ over relational objects.\hfill\mbox{}}\\
                	  \vspace{-0.13in}
                	\begin{algorithmic}[1]
                	    \If {match ($\groundinput$, $\exists o : r \cdot \phi'_f)$} \label{ln:extdiscovery} {\color{red}\textit{//process the existential operator}}
                	        \State $o' \gets \textsc{NewAct}(r)$  {\color{red}\textit{//create a new relational object of class $r$}}\label{ln:extnewClass}
                	        \State \Return $o'.ext \wedge \groundAlg$ ($\phi'_f$[$o \gets o'$], $\domainUnder$) \label{ln:extreplacement}
                	    \EndIf
                	    \If {match ($\groundinput$, $\forall o : r \cdot \phi'_f)$} \label{ln:unidiscovery} {\color{red}\textit{//process the universal operator }}
                	        \State \Return \begin{varwidth}[t]{\linewidth}
                	                    $\bigwedge_{[o':r] \in \domainUnder}$ $o'.ext\Rightarrow$ $\groundAlg$ ($\phi'_f$[$o \gets o'$], $\domainUnder$) \label{ln:forall}
                	                    \end{varwidth}
                	    \EndIf
                        \If {match ($\groundinput$, $\phi'_f \;op\; \psi'_f$ where $op = \wedge \mid \vee$)} {\Return $\groundAlg(\phi'_f, \domainUnder) \; op \; \groundAlg(\psi'_f, \domainUnder)$} \EndIf
                        \State \Return $\groundinput$ {\color{red}\textit{//case where $\groundinput$ is quantifier-free, including $\neg \phi'_f$ where $\phi'_f$ is atomic (NNF)}}
                	\end{algorithmic} 
                	\label{alg:ground}
                \end{algorithm}

\vspace{-0.15in}
\subsection{Counterexample-Guided Constraint Solving Algorithm} \label{subsec:main}
Let an \mfotl formula $\negp$ (to find a satisfiable counterexample to $P$), a set of \mfotl requirements $\reqs$, an optional volume bound $\bound$, and optionally a set of \tfol data domain constraints $\dataDomainConstraint$ be given. 
$\mainAlg$, shown in Alg.~\ref{alg:main}, searches for a solution $\sigma$ to $\negp \wedge \bigwedge_{\psi \in \reqs} \psi$
(with respect to $\dataDomainConstraint$) bounded by $\bound$, as a counter-example to $\bigwedge_{\psi \in \reqs} \psi \Rightarrow P$ (Def.~\ref{def:MFOTLMTC}). 
bounded by $\bound$. If no such solution is possible regardless of the bound, $\mainAlg$ returns UNSAT. If no solution can be found within the given bound, but a solution may exist for a larger bound, then $\mainAlg$ returns bounded-UNSAT. If $\bound$ is not specified, $\mainAlg$ will perform the search unboundedly until a solution or UNSAT is returned.

$\mainAlg$ first translates $\negp$ and every $\psi \in \reqs$ into \tfol formulas in $\reqs_f$, denoted by $\negpf$ and $\psi_f$, respectively. Then $\mainAlg$ searches for a satisfying solution to $\negpf \wedge \bigwedge_{\psi_f \in \reqs_f} \psi_f$ in the domain $\domain$ of volume, which is at most $\bound$.
Instead of searching in $\domain$ directly, $\mainAlg$ searches for a solution to $\negpf \wedge \bigwedge_{\psi_f \in \reqsUnder} \psi_f$ in $\domainUnder$ (denoted by $\underinput$) where $\reqsUnder \subseteq \reqs_f$ and $\domainUnder \subseteq \domain$.
$\mainAlg$ initializes $\reqsUnder$ and $\domainUnder$ as empty sets (LL:\ref{ln:initReq}-\ref{ln:initDomain}). Then, for the \tfol formula $\underinput$, $\mainAlg$ creates an over- and under-approximation query $\phi_g$ (L:\ref{ln:computeOverground}) and $\phi_g^{\bot}$ (L:\ref{ln:computeUnderground}), respectively (described in Sec.~\ref{subsec:OverandUnder}). 
$\mainAlg$ first solves the over-approximation query $\phi_g$ by querying an SMT solver (L:\ref{ln:overUNSAT}). If $\phi_g$ is unsatisfiable, then $\underinput$ is unsatisfiable (Lemma~\ref{lemma:overapprox}), and $\mainAlg$ returns UNSAT (L:\ref{ln:returnUNSAT}). \\
\indent
If $\phi_g$ is satisfiable, then $\mainAlg$ solves the under-approximation query $\phi_g^{\bot}$ (L:\ref{ln:solve2}). If $\phi_g^{\bot}$ is unsatisfiable, then the current domain $\domainUnder$ is too small, and $\mainAlg$ expands it (LL:\ref{ln:expansionstart}-\ref{ln:expansionfinish}).
This is because the satisfiability of $\phi_g$ indicates the possibility of finding a satisfying solution after adding at least one of the new relational objects in the solution to $\phi_g$ to $\domainUnder$ (Cor.~\ref{cor:DomainExpansion}).
The domain $\domainUnder$ is expanded by adding all relational objects $o'$ in the minimum (in terms of volume) solution $\sigma_{min}$ to $\phi_g$ (L:\ref{ln:minisolution}). 
To obtain $\sigma_{min}$, we follow MaxRes~\cite{Narodytska-et-al-14} methods: we analyze the UNSAT core of $\phi_g^{\bot}$ and incrementally weaken $\phi_g^{\bot}$ towards $\phi_g$ (i.e., the weakened query $\phi_g^{\bot'}$ is an ``over-under approximation'' that satisfies $\phi_g^{\bot} \Rightarrow \phi_g^{\bot'} \Rightarrow \phi_g$) until a satisfying solution $\sigma_{min}$ is obtained for the weakened query. 
However, if the volume of $\sigma_{min}$ exceeds $\bound$ (L:\ref{ln:exceedsLimit}), then bounded-UNSAT is returned (L:\ref{ln:returnboundedUNSAT}). UNSAT core-guided domain expansion has also been explored for unfolding the definition of recursive functions~\cite{DBLP:conf/cade/PassmoreCIABKKM20,DBLP:conf/sas/SuterKK11}. \\ 
\indent
On the other hand, if  $\phi_{g}^{\bot}$ yields a solution $\sigma$, then $\sigma$ is checked on $\reqs_f$ (L:\ref{ln:ReqCheck}). If $\sigma$ satisfies every $\psi_f$ in $\reqs_f$, then $\sigma$ is returned (L:\ref{ln:returnSAT}). If $\sigma$ violates some requirements in $\reqs_f$, then the violating requirement $\textit{lesson}$ is added to $\reqsUnder$ to be considered in the search for the next solutions (L:\ref{ln:ruleRefinement}). \\
\indent
If $\mainAlg$ does not find a solution or does not return UNSAT, it means that no solution is found because $\domainUnder$ is too small or $\reqsUnder$ are too weak. $\mainAlg$ then restarts with the expanded domain $\domainUnder$ or the refined set of requirements $\reqsUnder$. It computes the over- and under-approximation queries ($\phi_g$ and $\phi_g^{\bot})$ again, and repeats the steps.
See Sec.~\ref{subsec:example} for an illustration of $\mainAlg$. 

\begin{remark} 
\label{remark:opt}
$\mainAlg$ finds the optimal solution because it looks for the minimum solution $\sigma_{min}$ to the over-approximation query $\phi_g$ (L:\ref{ln:minisolution}) and uses it for domain expansion (L:\ref{ln:expandingdomain}). However, looking for $\sigma_{min}$ adds cost. If solution optimality is not required, $\mainAlg$ can be configured to heuristically find a solution $\sigma$ to $\phi_g$ such that $vol(\sigma) \le \bound$. The \textit{greedy best-first} search (gBFS) finds a solution to $\phi_g$ that minimizes the number of relational objects that are not already in $\domainUnder$, and then uses it to expand $\domainUnder$. 
We configured a non-optimal version of $\mainAlg$ ($\noopt$) that uses gBFS heuristics and evaluated its performance in Sec.~\ref{sec:evaluation}.
\end{remark}

\subsection{Illustration of $\mainAlg$}
\label{subsec:example}
\nf{Suppose a data collection centre (DCC) \textit{collect}s and \textit{access}es personal data information with two requirements: $\textit{req}_1$: data value can only be updated after having been collected or last updated for more than a week (168 hours); and $\textit{req}_2$:  data can only be accessed if has been collected or updated within a week (168 hours). 
The signature $S_{\textit{data}}$ for DCC contains three binary relations ($R_{\textit{data}}$): \textit{Collect}, \textit{Update}, and \textit{Access}, such that \textit{Collect}($d$, $v$), \textit{Update}($d$, $v$) and \textit{Access}($d$, $v$) hold at a given time point if and only if data at ID $d$ is collected, updated, and accessed with value $v$ at this time point, respectively. The \mfotl formulas for $P1$, $\textit{req}_1$ and $\textit{req}_2$ are shown in Fig.~\ref{fig:formulas}. Suppose $\mainAlg$ is invoked to find a counterexample for property $P1$ (shown in Fig.~\ref{fig:formulas}) subject to requirements $\reqs = \{\textit{req}_1, \textit{req}_2\}$ with the bound $\bound = 4$. $\mainAlg$ translates the requirements and the property to \tfol and initializes $\reqsUnder$ and $\domainUnder$ to empty sets. For each iteration, we use $\phi_g$ and $\phi_g^{\bot}$ to represent the over- and under-approximation queries computed on LL:\ref{ln:computeOverground}-\ref{ln:computeUnderground}, respectively.}

\nf{\underline{1st iteration:} $\domainUnder = \emptyset$ and $\reqsUnder = \emptyset$. Three new relational objects are introduced to $\phi_g$ (due to $\neg P1$): $\textit{access}_1$, $\textit{collect}_1$, and $\textit{update}_1$ such that: (C1) $\textit{access}_1$ occurs after $\textit{collect}_1$ and $\textit{update}_1$;(C2) $\textit{access}_1.d = \textit{collect}_1.d = \textit{update}_1.d$;(C3) $\textit{access}_1.v \neq \textit{collect}_1.v \wedge \textit{access}_1.v \neq \textit{update}_1.v$; and (C4) either $\textit{collect}_1$ or $\textit{update}_1$ must be in the solution. $\phi_g$ is satisfiable, but $\phi_{g}^{\bot}$ is UNSAT since $\domainUnder$ is an empty set. We assume $\domainUnder$ is expanded by adding $\textit{access}_1$ and $\textit{update}_1$. }

\nf{
\underline{2nd iteration:} $\domainUnder = \{\textit{access}_1, \textit{update}_1\}$ and $\reqsUnder = \emptyset$. The over-approximation $\phi_{g}$ stays the same, but $\phi_g^{\bot}$ becomes satisfiable since $\textit{access}_1$ and $\textit{update}_1$ are in $\domainUnder$. Suppose the solution is $\sigma_4$ (see Fig.~\ref{fig:traces}). However, $\sigma_4$ violates $\textit{req}_2$, so $\textit{req}_2$ is added to $\reqsUnder$.
}

\nf{
\underline{3rd iteration:} $\domainUnder = \{\textit{access}_1, \textit{update}_1\}$ and $\reqsUnder = \{\textit{req}_2\}$.  Two new relational objects are introduced in $\phi_g$ (due to $\textit{req}_2$): $\textit{collect}_2$ and $\textit{update}_2$ such that (C5) $\textit{collect}_2.time \le \textit{access}_1.time \le \textit{collect}_2.time + 168$; (C6) $\textit{update}_2.time \le \textit{access}_1.time \le \textit{update}_2.time + 168$; (C7) 
$\textit{access}_1.d = \textit{collect}_2.d = \textit{update}_2.d$; (C8) $\textit{access}_1.v = \textit{collect}_2.v = \textit{update}_2.v$; and (C9) $\textit{collect}_2$ or $\textit{update}_2$ is in the solution. The new $\phi_g$ is satisfiable, but $\phi_g^{\bot}$ is UNSAT because $\textit{update}_2 \not\in \domainUnder$ and $\textit{update}_1 \neq \textit{update}_2$ (C8 conflicts with C3). Therefore, $\domainUnder$ needs to be expanded.  Assume $\textit{collect}_2$ is added to $\domainUnder$.
}

\nf{
\underline{4th iteration:} $\domainUnder = \{\textit{access}_1, \textit{update}_1, \textit{collect}_2\}$ and $\reqsUnder = \{\textit{req}_2\}$. The over-approximation $\phi_g$ stays the same, but $\phi_g^{\bot}$ becomes satisfiable since $\textit{collect}_2$ is in $\domainUnder$. Suppose the solution is $\sigma_3$ (see Fig.~\ref{fig:traces}). Since $\sigma_3$ violates $\textit{req}_1$, $\textit{req}_1$ is added to $\reqsUnder$.
}

\nf{
\underline{5th iteration:} $\domainUnder = \{\textit{access}_1, \textit{update}_1, \textit{collect}_2\}$ and $\reqsUnder = \{\textit{req}_1, \textit{req}_2\}$.  The following constraints are added to $\phi_g$ (due to $\textit{req}_1$): (C9) $\neg (\textit{update}_2.time -168 \le \textit{collect}_1.time \le \textit{update}_2.time)$. Since (C9) conflicts with (C8), (C7) and (C1), $\textit{update}_2$ cannot be in the solution to $\phi_g$. The over-approximation $\phi_g$ is satisfiable if $\textit{collect}_1$ (introduced in the 1st iteration) or $\textit{update}_2$ (3rd iteration) are in the solution. However, $\phi_g^{\bot}$ is UNSAT since $\domainUnder$ does not contain $\textit{collect}_1$ or $\textit{update}_2$. Thus, $\domainUnder$ is expanded.  
Assume $\textit{update}_2$ is added to $\domainUnder$.
}

\nf{
\underline{6th iteration:} $\domainUnder = \{\textit{access}_1, \textit{update}_1, \textit{collect}_2, \textit{update}_2\}$,  $\reqsUnder = \{\textit{req}_1, \textit{req}_2\}$. The following constraints are added to $\phi_g$ (C10) $\textit{update}_2.time \ge \textit{update}_1.time + 168$ (due to $req_1$) and (C11) $\textit{update}_2.time \le $  $\textit{update}_1.time$ (due to $\neg P$). Since (C10) conflicts with (C11), $\textit{update}_2$ cannot be in the solution to $\phi_g$. Thus, $\phi_g$ is satisfiable only if $\textit{collect}_1$ is in the solution. However, $\phi_g^{\bot}$ is UNSAT because $\textit{collect}_1 \not\in \domainUnder$. Therefore,  $\domainUnder$ is expanded by adding $\textit{collect}_1$. 
}

\nf{
\underline{final iteration:} $\domainUnder = \{\textit{access}_1, \textit{update}_1, \textit{collect}_2, \textit{update}_2,$ $\textit{collect}_1\}$ and $\reqsUnder = \{\textit{req}_1, \textit{req}_2\}$. The under-approximation $\phi_{g}^{\bot}$ becomes satisfiable, and yields the solution $\sigma_5$ in Fig.~\ref{fig:traces} which satisfies both $\textit{req}_1$ and $\textit{req}_2$.
}
\vspace{-0.15in}
\section{Evaluation}
\label{sec:evaluation}
To evaluate our approach, we developed a prototype tool, \nf{called {\sffamily LEGOS}}, that implements our \mfotl bounded satisfiability checking algorithm, $\mainAlg$ (Alg.~\ref{alg:main}).
It includes Python API for specifying system requirements and \mfotl safety properties. 
We use pySMT~\cite{pysmt2015} to formulate SMT queries and Z3~\cite{Leonardo-Nikolaj-2008} to check their satisfiability. The implementation and the evaluation artifacts are included in the supplementary material~\cite{tacas-supp}. 
In this section, we evaluate the effectiveness of our approach using five case studies, aiming to answer the following research question: \emph{How effective is our approach at determining the bounded satisfiability of \mfotl formulas?}
We measure effectiveness in terms of the ability to determine satisfiability (i.e., the satisfying solution and its volume, UNSAT, or bounded UNSAT), and performance, i.e., time and memory usage.

\vskip 0.05in
\noindent
{\bf Cases studies.}
The five case studies considered in this paper are summarized below:
(1) PHIM (derived from~\cite{Feng-Marsso-Garavel-21,Arfelt-Basin-Debois-19}): a computer system for keeping track of personal health information with cost management; 
(2) CF@H\footnote[1]{\url{https://covidfreeathome.org/}}: a system for monitoring COVID patients at home and enabling doctors to monitor patient data; 
(3) PBC~\cite{Basin-15}: an approval policy for publishing business reports within a company; 
(4) BST~\cite{Basin-15}: a banking system that processes customer transactions; and  (5)~NASA~\cite{Mattarei-15}: an automated air-traffic control system design that aims to avoid aircraft collisions. \footnote[2]{The requirements and properties for the NASA case study are originally expressed in LTL, which is subsumed by \mfotl.}  
Tbl.~\ref{tab:stats} gives their statistics.   For each case study, we record the number of requirements, relations, relation arguments, and properties, denoted as $\#reqs$, $\#rels$, $\#args$, and $\#props$, respectively. Additionally, Tbl.~\ref{tab:stats} shows initial configurations used in our experiments, with
 number of custodians ($\#c$), patients ($\#p$), and data ($\#d$) for PHIM; number of users ($\#u$), and data ($\#d$) for CF@H and PBC; number of employees ($\#e$), customers ($\#c$), transactions ($\#t$), and the maximum amount for a transaction ($\textit{sup}$) for BST; number of ground-separated  ($\#GSEP$) and of the self-separating aircraft ($\#SSEP$) for NASA. 

\begin{minipage}{0.37\linewidth}
\begin{table}[H]
    \centering
            \scalebox{0.585}{
                \begin{tabular}{|c|c|c|c|c|c|}
                    \hline
                     \multirow{2}{*}{names} & \multicolumn{4}{|c|}{case study statistics} & \multirow{2}{*}{configuration} \\\cline{2-5}
                     & $\#reqs$ & $\#rels$ & $\#args$ & $\#props$ &  \\ \hline
                     \multirow{2}{*}{PHIM} & \multirow{2}{*}{18} & \multirow{2}{*}{22} & \multirow{2}{*}{$[1-4]$} & \multirow{2}{*}{6} & $\textit{\#c}=2$, $\textit{\#p}=2$\\
                         &    &   &         &   & $\textit{\#d}=5$
                     \\ \hline
                     CF@H & 45 & 28 & $[2-3]$ & 7 & $\textit{\#u}=2$, $\textit{\#d}=10$
                     \\ \hline
                     PBC & 14 & 7 & $[1-2]$ & 1 & $\textit{\#u}=5$, $\textit{\#d}=10$
                     \\ \hline
                     \multirow{2}{*}{BST} & \multirow{2}{*}{10} & \multirow{2}{*}{3} & \multirow{2}{*}{$[1-3]$} & \multirow{2}{*}{3} & $\textit{\#e}=1$, $\textit{\#c}=2$\\
                         &    &   &         &   & $\textit{\#t}=4$, $\textit{sup}=10$
                     \\ \hline
                     \multirow{4}{*}{NASA} & \multirow{4}{*}{194} & \multirow{4}{*}{10} & \multirow{4}{*}{$[6-79]$} & \multirow{4}{*}{6} & $\#GSEP=3$ \\
                         &    &   &         &   & $\#SSEP=0$\\
                         &    &   &         &   & $\#GSEP=2$\\
                         &    &   &         &   & $\#SSEP=2$
                     \\ \hline
                \end{tabular}\label{tab:models}
            }
    \caption{{\footnotesize Case study statistics. 
    }}
    \label{tab:stats}
\end{table}
\end{minipage} ~~
\begin{minipage}{0.6\linewidth}
\begin{table}[H]
    \centering
        \scalebox{0.68}{
            \begin{tabular}{|c||c|c|c|c|c|c|c|c|c|c|c|c|c|}
                 \toprule
                 NASA & \multicolumn{6}{|c|}{configuration 1} & & \multicolumn{6}{|c |}{configuration 2} \\
                 \cline{2-7}\cline{9-14}
                 & \multicolumn{3}{|c|}{$\mainAlg$} & \multicolumn{3}{|c|}{nuXmv} & & \multicolumn{3}{|c |}{$\mainAlg$} &  \multicolumn{3}{|c|}{nuXmv} \\
                 \cline{2-7}\cline{9-14}
                  & \multirow{2}{*}{out.} & time &   mem. & \multirow{2}{*}{out.} & time & mem.  ~&  & \multirow{2}{*}{out.} & time &  mem. &  \multirow{2}{*}{out.} & time & mem. \\ 
                   & & (sec) & (MB) &  & (sec) & (MB)  &~& & (sec) & (MB) &  & (sec) & (MB) \\
                  \cmidrule[1.5pt]{0-6}\cmidrule[1.5pt]{9-14}
                   $na_1$ & U& 0.80  & 154  &  U &  0.88 & 82 & & U & 0.13  & 141 & U & 1.65 & 90 \\\cline{2-7}\cline{9-14}
                   $na_2$ & U& 0.16 & 141 & U &  0.47 & 70 &  & U & 0.15 & 141 & U& 1.50 & 90 \\\cline{2-7}\cline{9-14}
                   $na_3$ & U& 0.16 & 141 & U & 0.49 & 83 &  & U & 0.13 & 141 & U& 1.48 & 90 \\\cline{2-7}\cline{9-14}
                   $na_4$ & U& 0.77 & 80  & U & 0.54 & 83 &  & U & 0.15 & 66 & U & 1.43 & 91 \\\cline{2-7}\cline{9-14}
                   $na_5$ & U& 0.14 & 140 & U & 0.52 & 82 & & U & 0.15 & 141 & U& 1.43 & 90 \\\cline{2-7}\cline{9-14}
                   $na_6$ & U& 0.03 & 62  & U & 0.57 & 72  & & U & 0.03 & 62 & U & 1.40 & 90
                 \\\cline{2-7}\cline{9-14}
                 \bottomrule
            \end{tabular} 
            } 
    \caption{{\footnotesize Performance comparison between  $\mainAlg$ and nuXmv on case study NASA. 
    }}
    \label{tab:effectiveness-2}
\end{table}
\end{minipage}

Case studies were selected for (i) the purpose of comparison with existing works (i.e., NASA); (ii) checking whether our approach scales with case studies involving data/time constraints (PBC, BST, PHIM and CF@H); or (iii) evaluating the applicability of our approach with real-word case studies (CF@H and NASA). 
In addition to prior case studies, we include PHIM and CF@H which have complex data/time constraints.
The number of requirements for the five case studies ranges  between ten (BST) and 194 (NASA).
The number of relations present in the \mfotl requirements ranges from three (BST) to 28 (CF@H), and the number of arguments in these relations ranges from 1 (PHM, PBC, and BST) to 79 (NASA).

\vskip 0.05in
\noindent
{\bf Experimental setup.}
Given a set of requirements, data constraints and properties of interest for each case study,
we measured the run-time (time) and peak memory usage (mem.) of performing bounded satisfiability checking of \mfotl properties, and the volume $vol_\sigma$ (the number of relational objects) of the solution ($\sigma$) with ($\opt$) and without ($\noopt$) the optimality guarantees (see Remark~\ref{remark:opt} for finding non-optimal solutions). 
We conduct two experiments: the first one evaluates the efficiency and scalability of our approach; the second one compares our approach with satisfiability checking.
Since there is no existing work for checking \mfotl satisfiability, we compared with LTL satisfiability checking because \mfotl subsumes LTL.
To study the scalability of our approach, our first experiment considers four different configurations obtained by increasing the data constraints of the case-study requirements.
The initial configuration (small) is described in Tbl.~\ref{tab:stats} and the initial bound is 10. The medium and large configurations are obtained by multiplying the initial data constraints and volume bound by ten and hundred, respectively. The last (unbounded) configuration does not bound either the data domain or the volume. 
\nf{As we noted earlier in Sec.~\ref{sec:approach}, the purpose of adding data constraints is to avoid unrealistic counterexamples. For example, the NASA case study uses a data set for specifying the possible system control modes and uses data ranges to restrict the possible measures from the aircraft (e.g., aircraft's trajectory). In the other case studies, data constraints are realistic data ranges (e.g., a patient's account balance should be non-negative).}
To study the performance of our approach relative to existing work, our second experiment considers two configurations of the NASA case study verified in \cite{Li-et-al-19} using the state-of-the-art symbolic model checker nuXmv~\cite{DBLP:conf/cav/CavadaCDGMMMRT14}\footnote{\nf{{\sffamily LEGOS} solved all configurations from the NASA case study; see the results in \cite{tacas-supp}. For comparison, we report only on the configurations that are explicitly supported by nuXmv.}}. 
We compare our approach's result against the reproduced result of nuXmv verification. 
For both experiments, we report the analysis outcomes, i.e., the volume of the satisfying solution (if one exists), UNSAT, or bounded UNSAT; and performance, i.e., time and memory usage. 
The experiments were conducted using a ThinkPad X1 Carbon with an Intel Core i7 1.80 GHz processor, 8 GB of RAM, and running 64-bit Ubuntu GNU/Linux 8. 

\begin{table}[t]
    \centering
            \scalebox{0.6}{
            \begin{tabular}{|c c||c|c|c|c|c|c|c|c|c|c|c|c|c|c|c|}
                 \toprule
                  \multicolumn{2}{|c}{case studies} & \multicolumn{3}{|c|}{small} & & \multicolumn{3}{|c|}{medium} & & \multicolumn{3}{|c |}{big} & & \multicolumn{3}{|c|}{unbounded} \\\cline{3-5}\cline{7-9}\cline{11-13}\cline{15-17}
                  \multicolumn{2}{|c |}{} & \multirow{2}{*}{out.} & time &   mem. &~ & \multirow{2}{*}{out.} & time & mem.  ~&  & \multirow{2}{*}{out.} & time &  mem. &  ~& \multirow{2}{*}{out.} & time & mem. \\ 
                  \multicolumn{2}{|c |}{} & & (sec) & (MB) &~ &  & (sec) & (MB)  &~& & (sec) & (MB) &  ~& & (sec) & (MB) \\
                  \cline{3-5}\cline{7-9}\cline{11-13}\cline{15-17}
                  \multicolumn{2}{|c |}{} & $\noopt$ $|$ $\opt$  & $\noopt$  $|$ $\opt$  &  $\noopt$  $|$ $\opt$  & & $\noopt$  $|$ $\opt$  & $\noopt$  $|$ $\opt$  & $\noopt$  $|$ $\opt$  &  & $\noopt$  $|$ $\opt$  & $\noopt$  $|$ $\opt$  & $\noopt$  $|$ $\opt$  &  & $\noopt$  $|$ $\opt$  & $\noopt$  $|$ $\opt$  & $\noopt$  $|$ $\opt$  \\\cmidrule[1.5pt]{1-3}\cmidrule[1.5pt]{3-5}\cmidrule[1.5pt]{7-9}\cmidrule[1.5pt]{11-13}\cmidrule[1.5pt]{15-17}
                 \multirow{7}{*}{PHIM}& $ph_1$ & U & 0.04 $|$ 0.03 & 29 $|$ 29 & & U& 0.03 $|$ 0.03 & 136 $|$ 136 & & U& 0.04 $|$ 0.04 & 136 $|$ 136 & & U& 0.06 $|$ 0.05 & 64 $|$ 64 \\\cline{3-5}\cline{7-9}\cline{11-13}\cline{15-17}
                 & $ph_2$ & U & 0.03 $|$ 0.03 & 138 $|$ 138 & & U& 0.03 $|$ 0.03 & 136 $|$ 137 & & U& 0.03 $|$ 0.04 & 136 $|$ 136 & & U& 0.05 $|$ 0.06 & 64 $|$ 61 \\\cline{3-5}\cline{7-9}\cline{11-13}\cline{15-17}
                  & $ph_3$ & U& 0.03 $|$ 0.03 & 134 $|$ 137 & & U& 0.03 $|$ 0.03 & 138 $|$ 138 & & U& 0.05 $|$ 0.05 & 137 $|$ 138 & & U& 0.06 $|$ 0.06 & 64 $|$ 64 \\\cline{3-5}\cline{7-9}\cline{11-13}\cline{15-17}
                 & $ph_4$ & U& 0.04 $|$ 0.04 & 136 $|$ 138 & & U& 0.04 $|$ 0.04 & 138 $|$ 135 &  & U& 0.05 $|$  0.05 & 138 $|$ 138 & & U& 0.06 $|$ 0.07 & 64 $|$ 64\\\cline{3-5}\cline{7-9}\cline{11-13}\cline{15-17}
                 & $ph_5$ & U& 0.02 $|$ 0.02 & 135 $|$ 135 & & U& 0.02 $|$ 0.02 & 608  $|$ 608 & & 56 $|$ 56 & 30.51 $|$ 30.51 & 390 $|$ 390 & & 56 $|$ 56 & 21.64 $|$ 21.60 & 393 $|$ 390 \\\cline{3-5}\cline{7-9}\cline{11-13}\cline{15-17}
                 & $ph_6$ & b-U& 0.18 $|$ 0.20 & 139 $|$ 139 & & U& 0.72 $|$ 0.82 & 144 $|$ 144 & & U& 0.88 $|$ 0.70 & 142 $|$ 142 & & U& 0.91  $|$ 0.91 & 70 $|$ 70 \\\cline{3-5}\cline{7-9}\cline{11-13}\cline{15-17}
                 & $ph_7$ & U & 0.11 $|$ 0.11 & 139 $|$ 139 & & 29 $|$ 29 & 13.80 $|$ 1905.40 & 193 $|$ 599 & & \textbf{30 $|$ 29} & 20.25 $|$ 682.22 & 193 $|$ 601 & & \textbf{32  $|$ 29} & 20.96  $|$ 1035.87 & 123 $|$ 383 \\
                 \cmidrule[1.5pt]{1-3}\cmidrule[1.5pt]{3-5}\cmidrule[1.5pt]{7-9}\cmidrule[1.5pt]{11-13}\cmidrule[1.5pt]{15-17}
                  \multirow{7}{*}{CF@H} & $cf_1$ &  b-U & 4.80 $|$ 6.90 & 114 $|$ 176 & & U& 2.87 $|$ 3.55 & 81 $|$ 86 & & U& 2.98 $|$ 1.71 & 85 $|$ 76 & & U& 1.71 $|$ 0.74 & 74 $|$ 68 \\\cline{3-5}\cline{7-9}\cline{11-13}\cline{15-17}
                  & $cf_2$ & b-U & 0.87 $|$ 0.93 & 70 $|$ 70 & & 14 $|$ 14 & 3.21 $|$ 425.41 & 79 $|$ 334 & & 14 $|$ 14 & 2.40 $|$ 778.36 & 76 $|$ 80 & & 14 $|$ 14 & 3.32 $|$ 16.97 & 80 $|$ 205 \\\cline{3-5}\cline{7-9}\cline{11-13}\cline{15-17}
                  & $cf_{3}$ & b-U & 1.38 $|$ 1.31 & 145 $|$ 145 & & 16 $|$ 16 & 6.05 $|$ 90.78 & 168 $|$ 403 & & 16 $|$ 16 & 3.54 $|$ 371.65 & 157 $|$ 846 & & 16 $|$ 16 & 5.35 $|$ 24.07 & 86 $|$ 164
                 \\\cline{3-5}\cline{7-9}\cline{11-13}\cline{15-17}
                  & $cf_{4}$ & b-U & 1.52 $|$ 0.73 & 74 $|$ 68 & & 14 $|$ 14 & 4.54  $|$ 65.59 & 90  $|$ 261 & & 14 $|$ 14 & 5.63  $|$ 57.30 & 95 $|$ 261 & & 14 $|$ 14 & 5.65  $|$ 1227.02 & 89 $|$ 294 
                 \\\cline{3-5}\cline{7-9}\cline{11-13}\cline{15-17}
                  & $cf_{5}$ &  8 $|$ 8 & 1.20 $|$ 1.17  & 146 $|$ 147 & & 8 $|$ 8 & 0.48 $|$ 0.54 & 141 $|$ 142 & & 8 $|$ 8 & 0.69 $|$ 0.57 & 141 $|$ 141 & & 8 $|$ 8 & 0.72 $|$ 0.76 & 69 $|$ 69 
                 \\\cline{3-5}\cline{7-9}\cline{11-13}\cline{15-17}
                  & $cf_{6}$ &  8 $|$ 8 & 1.06 $|$ 1.16 & 146 $|$ 147 & &  8 $|$ 8 & 0.52 $|$ 0.61 & 142 $|$ 142 & & 8 $|$ 8 & 0.60 $|$ 0.73 & 141 $|$ 141 & & 8 $|$ 8 & 0.72 $|$ 0.72 & 69 $|$ 69
                 \\\cline{3-5}\cline{7-9}\cline{11-13}\cline{15-17}
                  & $cf_{7}$ & U & 0.58 $|$ 0.58 & 141 $|$ 142 & & U & 0.38 $|$ 0.36 & 140 $|$ 141 & & U& 0.47 $|$ 0.44 & 140 $|$ 141 & & U& 0.30 $|$ 0.34 & 66 $|$ 67
                 \\\cmidrule[1.5pt]{1-3}\cmidrule[1.5pt]{3-5}\cmidrule[1.5pt]{7-9}\cmidrule[1.5pt]{11-13}\cmidrule[1.5pt]{15-17}
                 PBC & $pb_1$ & U& 0.04 $|$ 0.04 & 29 $|$ 140 & & U& 0.16 $|$ 0.17 & 140 $|$ 139 & & 9 $|$ 9 & 0.28 $|$ 0.29 & 141 $|$ 141 & & 9 $|$ 9 & 0.27 $|$ 0.28  & 67 $|$ 67\\
                 \cmidrule[1.5pt]{1-3}\cmidrule[1.5pt]{3-5}\cmidrule[1.5pt]{7-9}\cmidrule[1.5pt]{11-13}\cmidrule[1.5pt]{15-17}
                 \multirow{3}{*}{BST} & $bs_1$ & U & 0.04 $|$ 0.03  & 64 $|$ 63  & & U & 0.29  $|$ 0.24   & 70  $|$ 68 & & U & 0.31  $|$ 0.30  & 69  $|$ 68  & & U & 0.25 $|$ 0.25 & 69 $|$ 69 
                 \\\cline{3-5}\cline{7-9}\cline{11-13}\cline{15-17}
                 & $bs_2$ & 2 $|$ 2 & 0.04 $|$ 0.04 & 62 $|$ 64 & & 2 $|$ 2 & 0.04 $|$ 0.04 & 62 $|$ 62 & & 2 $|$ 2 & 0.04 $|$ 0.04 & 64 $|$ 64 & & 2 $|$ 2 & 0.04 $|$ 0.04 & 64$|$ 64
                 \\\cline{3-5}\cline{7-9}\cline{11-13}\cline{15-17}
                 & $bs_3$ & U& 0.02 $|$ 0.02 & 62 $|$ 62 & & 5 $|$ 5 & 0.4 $|$ 0.9 & 70 $|$ 73 & & 5 $|$ 5 & 0.39 $|$ 0.85 & 70 $|$ 74 & & 5 $|$ 5 & 0.40 $|$0.70 & 70 $|$ 72
                 \\\cmidrule[1.5pt]{1-3}\cmidrule[1.5pt]{3-5}\cmidrule[1.5pt]{7-9}\cmidrule[1.5pt]{11-13}\cmidrule[1.5pt]{15-17}
            \end{tabular}
            } 
    \caption{
    {\footnotesize 
    Run-time performance for four case studies and 18 properties. We record the outcome (out.) of the algorithm with ($\opt$) or without ($\noopt$) the optimal solution guarantee: UNSAT (U), bounded-UNSAT (b-U), or the volume of the counterexample $\sigma$ (a natural number, corresponding to vol$_\sigma$). We consider four different configurations: small (see Tab.~\ref{tab:models}), medium (x$10$), big (x$100$), and unbounded ($\infty$) data domain constraints and volume bound. 
    Volume differences between $\opt$ and $\noopt$ are bolded. 
     }\vspace*{-2em}}
    \label{tab:evalc}
\end{table}

\vskip 0.05in
\noindent
{\bf Results of the first experiment} are summarized in Tbl.~\ref{tab:evalc}. 
Out of the 72 trials, our approach found 31 solutions.  It also returned five bounded-UNSAT answers, and 36 UNSAT answers.
The results show that our approach is effective in checking satisfiability of case studies with different sizes. 
More precisely, we observe that it takes under three seconds to return UNSAT and between .04 seconds ($bs_2$:medium) and 32 minutes ($ph_7$:medium:$\opt$) to return a solution.
In the worst case, $\opt$ took 32 minutes for checking $ph_7$ where the property and requirements contain complex constraints.   
Effectively, $ph_7$ requires the deletion of data stored at id 10, while the cost of deletion increases over time under PHIM's requirements.  
Therefore, the user has to perform a number of actions to obtain a sufficient balance to delete the data. Additionally, each action that increases the user's balance has its own preconditions, effects, and time cost, making the process of choosing the sequence of actions to meet the increasing deletion cost non-trivial.

We can see a difference in time between cf2 `big' and `unbounded', this is because the domain expansion followed two different paths and one produces significantly easier SMT queries. Since our approach is guided by counterexamples (i.e., the path is guided by the solution from the SMT solver (Alg.1-L:\ref{ln:minisolution})), our approach does not have direct control over the exact path selection. In future work, we aim to add optimizations to avoid/backtrack from hard paths.

We observe that the data-domain constraint and volume bound used in different configurations do not affect the performance of $\mainAlg$ when the satisfiability of the instances does not depend on them, which is the case for all the instances except for $ph_{6-7}$:small, $cf_{1-3}$:small, and $bs_3$:small. 
As mentioned in Sec.~\ref{sec:approach}, the data-domain constraint ensures that satisfying solutions have realistic data values. For $ph1-ph4$, the bound used in the small, medium and large configurations creates additional constraints in the SMT queries for each relational object, and therefore results in a larger peak memory than the unbounded configuration.

Finding the optimal solution (by $\opt$), in contrast to finding a satisfying solution without the optimal guarantee (by $\noopt$), imposes a substantial computational cost while rarely achieving a volume reduction. The non-optimal heuristic $\noopt$ often outperformed the optimal approach for satisfiable instances. Out of 31 satisfiable instances, $\noopt$ solved 12 instances 3 times faster, 10 instances 10 times faster and seven instances 20 times faster than $\opt$. 
Compared to the non-optimal solution, the optimal solution reduced the volume for only two instances: $ph_7$:large and  $ph_7$:unbounded by one (3\%) and three (9\%), respectively.
On all other satisfying instances, $\opt$ and $\noopt$ both find the optimal solutions. When there is no solution, both $\opt$ and $\noopt$ are equally efficient.

\vskip 0.05in
\noindent
{\bf Results of the second experiment}
are summarized in Tbl.~\ref{tab:effectiveness-2}. Our approach and nuXmv both correctly verified that all six properties were UNSAT in both NASA configurations. 
We observe that the performance of our approach is comparable to nuXmv for the first configuration with .10 to .20 seconds of difference on average. Yet, for the second configuration, our approach terminates in less than 0.20 sec and nuXmv takes 1.50 seconds on average. 
We conclude that our approach's performance is comparable to that of nuXmv for LTL satisfiability checking even though our approach is not specifically designed for LTL. 

\vskip 0.05in
\noindent
{\bf Summary.}
In summary, 
we have demonstrated that our approach is effective at determining the bounded satisfiability of \mfotl formulas using case studies with different sizes and from different application domains. 
When restricted to LTL, our approach is at least as effective as the existing work on LTL satisfiability checking which uses a state-of-the-art symbolic model checker. 
Importantly, $\mainAlg$ can often determine satisfiability of instances without reaching the volume bound, and its performance is not sensitive to the data domain. 
On the other hand, $\mainAlg$'s optimal guarantee imposes a substantial computational cost while rarely achieving a volume reduction over non-optimal solutions obtained by $\noopt$. 
We need to investigate the trade-off between optimality and efficiency, as well as evaluate the performance of $\mainAlg$ on a broader range of benchmarks.

\section{Related Work}
\label{sec:relatedwork}

Below, we compare with the existing approaches that address the satisfiability checking of temporal logic and first-order logic. 

%

%
%

\vskip 0.05in
\noindent
{\bf Satisfiability checking of temporal properties.} Temporal logic satisfiability checking has been studied for the verification of system  designs. Satisfiability checking for Linear Temporal Logic (LTL) can be performed by reducing the problem  to model checking~\cite{Rozier-et-al-07}, by applying automata-based techniques~\cite{Li-et-al-13}, or by SAT solving~\cite{Li-et-al-20,DBLP:conf/aaai/LiRPZV19,DBLP:journals/logcom/LiP0VH18,DBLP:journals/japll/BersaniFMPRP14}.  Satisfiability checking for metric temporal logic (MTL)~\cite{DBLP:journals/tosem/PradellaMP13} and its variants, e.g., mission-time LTL~\cite{Li-et-al-19} and signal temporal logic~\cite{DBLP:journals/pacmpl/BaeL19}, has been studied 
for the verification of real-time system designs. These existing techniques are inadequate for our needs: LTL and MTL cannot effectively capture quantified data constraints commonly used in legal properties.  MFOTL does not have such a limitation as it extends MTL and LTL with first-order quantifiers, thereby supporting the specification of data constraints.
\vskip 0.05in
\noindent
{\bf Finite model finding for first-order logic.} 
Finite-model finders~\cite{claessen2003new,DBLP:conf/cade/ReynoldsTGKDB13} look for a model by checking universal quantifiers exhaustively over candidate models with progressively larger domains; we look for finite-volume solutions using
a similar approach.  
On the other hand,
we consider an explicit bound on the volume of the solution, and are able to find the solution with the smallest volume. SMT solvers support quantifiers with quantifier instantiation heuristics~\cite{DBLP:conf/cav/GeM09,DBLP:conf/cade/GeBT07} such as E-matching ~\cite{DBLP:journals/jacm/DetlefsNS05,DBLP:conf/cade/MouraB07} and conflict-based instantiation~\cite{DBLP:conf/fmcad/ReynoldsTM14}. Quantifier instantiation heuristics are nonetheless generally incomplete, whereas, in our approach, we obtain completeness by bounding the volume of the satisfying solution.

\section{Conclusion}
\label{sec:conclusion}
In this paper, we proposed an incremental bounded satisfiability checking approach, called $\mainAlg$, aimed to enable verification of legal properties, expressed in \mfotl, against system requirements.
$\mainAlg$ first translates \mfotl formulas to first-order logic with relational objects (\tfol) and then searches for a satisfying solution to the translated \tfol formulas in a bounded search space by deriving over- and under-approximating SMT queries.  
$\mainAlg$ starts with a small search space and incrementally expands it until an answer is returned or until the bound is exceeded.
We implemented $\mainAlg$ on top of the SMT solver Z3.  Experiments using five case studies  showed that our approach is 
effective for identifying errors in requirements from different application domains.
%
Our approach is currently limited to verifying safety properties. In the future, we plan to extend our approach so that it can handle a broader spectrum of  property types, including liveness and fairness.
$\mainAlg$'s performance and scalability depend crucially on how the domain of relational objects is maintained and expanded. As future work, we would like to study the effectiveness of other heuristics to improve $\mainAlg$'s scalability (e.g., random restart and expansion with domain-specific heuristics). We also aim to study how to learn/infer \mfotl properties during search to further improve the efficiency of our approach.

\bibliographystyle{splncs04}
\bibliography{reference}

\clearpage
\appendix

\section*{Appendix}
Sec.~\ref{appendix:over} provides the correctness proof for the constructions of over- and under-approximation queries; Sec.~\ref{ap:SoundTerminateOpt} 
studies its correctness (Th.~\ref{thm:soundness}), termination (Th.~\ref{thm:termination}) and optimality (Th.~\ref{Thm:optimality}).

\section{Correctness Proof of Over- and Under- Approximation}
\label{appendix:over}
In this section, we prove the correctness of the over and under-approximation (Lemma~\ref{lemma:overapprox} and Lemma~\ref{lemma:underapprox}).

\begin{proposition} \label{prop:QFfree}
For every \tfol formula $\phi_f$ and domain $\domainUnder$,  the grounded formula $\phi_g = \groundAlg (\phi_f, \domainUnder)$ is quantifier-free and contains a finite number of variables and terms.
\end{proposition}

\begin{proof}
We note that (1) quantifiers are limited only to relational objects for \tfol formula $\phi_f$, and they are eliminated by $\groundAlg$; (2) since the number of a relational objects in the domain $\domainUnder$ is finite, each $\forall$ is expanded into conjunctions of a finite number of terms; (3) finally, since the formula $\phi_f$ is finite and does not contain cyclic reference, the number of times that $\groundAlg$ is invoked during $\groundAlg(\phi_f)$ is always finite. Combining (1), (2) and (3), we obtain that $\phi_g$ is quantifier-free and contains a finite number of variables and terms. 
\end{proof}

We now present proof of correctness for the over-approximation (Lemma~\ref{lemma:overapprox})
\begin{proof} [Proof of Lemma~\ref{lemma:overapprox}]
Suppose $\phi_g$ is UNSAT but there exists a solution $v_f$ for $\phi_f$ in some domain $\domain$ ($\domain$ may be different from $\domainUnder$). We show that we can always construct a solution $v_g$ that satisfies $\phi_g$, which causes a contradiction. First, we construct a solution $v_g'$ for $\phi_g' = \groundAlg (\phi_f, \domain)$ from the solution $v_f$ (for $\phi_f$). Then, we construct a solution $v_g$ for $\phi_g$ from the solution $v_g'$ for $\phi_g'$.

We can construct a solution $v_g'$ for $\phi_g'$ in $\domain \cup NewRs$ where $NewRs$ are the new relational objects added by $\groundAlg$. The encoding of $\groundAlg$ uses the standard way for grounding universally quantified expression by enumerating every relational object in $\domain$ (L:\ref{ln:forall}). For every existentially quantified expression, there exists some relation object $o \in \domain$ enabled by $v_f$ (i.e., $v_f(o.ext) = \top$) that satisfies the expression in $\phi_f$, whereas $\phi_g'$ contains a new relational object $o' \in NewRs$ for satisfying the same expression (L:\ref{ln:extreplacement}). Let $v_f(o) = v_g'(o')$ for $o$ and $o'$, and then $v_g'$ is a solution to $\phi_g'$.

To construct the solution $v_g$ for $\phi_g = \groundAlg(\phi_f, \domainUnder)$ from the solution $v_g'$ for $\phi_g' = \groundAlg (\phi_f, \domain)$, we consider expansion of the universally quantified expression in $\phi_f$ (L:\ref{ln:unidiscovery}). For every relational objects in $o^{+} \in \domain \setminus \domainUnder$, $\groundAlg$ creates constraints (L:\ref{ln:forall}) in $\phi_g'$, but not in $\phi_g$. On the other hand, for every relational object in $o^{-} \in \domainUnder \setminus \domain$, we disable $o^{-}$ in the solution $v_g$ by assigning $o_g(r^{-}.ext) \gets \bot$. Therefore, the constraints instantiated by $o^{-}$ (at L:\ref{ln:forall}) in $\phi_g$ are vacuously satisfied. 

For every relational object $o \in \domainUnder \cap \domain$, we let $v_g(o) = v_g'(o)$, and all shared constraints in $\phi_g$ and $\phi_g'$ are satisfied by $v_g$ and $v_g'$, respectively.  Therefore, $v_g$ is a solution to $\phi_g$. Contradiction.
\end{proof}

We now present proof of correctness for the over-approximation (Lemma~\ref{lemma:underapprox})
\begin{proof}[Proof of Lemma~\ref{lemma:underapprox}]
If $\sigma$ is a solution to $\phi_g^{\bot}$ in the domain  $\domainUnder \cup  NewRs$, then we can construct a solution $\sigma'$ to $\phi_f$ in the domain $\domainUnder$. The construction of $\sigma'$ simply ignores any relational object in $\sigma$ that does not appear in $\domainUnder$ (i.e., the ones in $NewRs$). The solution $\sigma'$ is valid for $\phi_f$ in $\domainUnder$ because for every ignored relational object $o$,   $\textsc{NoNewR}(NewRs, \domainUnder)$ guarantees that some relational object $o' \in \domainUnder$ is semantically equivalent to $o$. Therefore, if an existentially quantified expression is satisfied by $o$,  it is also satisfied by $o'$. On the other hand, universally quantified expression in $\phi_g^{\bot}$ are grounded by considering only $\domainUnder$  (L:\ref{ln:forall} of Alg.~\ref{alg:ground}), and hence $\sigma'$ satisfies them. Therefore, $\sigma'$ is a solution to $\phi_f$ in $\domainUnder$. 
\end{proof}

\section{Correctness, Termination, Optimality of \mainAlg}
\label{ap:SoundTerminateOpt}
In this section, we prove that algorithm $\mainAlg$ is correct and optimal, i.e., always finds a solution with a minimum volume.  We also show that $\mainAlg$ terminates.  

 \begin{theorem} [Soundness] \label{thm:soundness}
 If the algorithm $\mainAlg$ terminates on input $P$, $\reqs$ and $\bound$, then it returns the correct result, i.e., a counter-example $\sigma$,  ``UNSAT'' or ``bounded-UNSAT'', when they apply. 
 \end{theorem}
 
 \noindent
   \textbf{Proof.}
    Let $\phi_f$ be the \tfol formula $T(\neg P)\bigwedge_{\psi \in \reqs} T(\psi)$.   We consider correctness of $\mainAlg$ for three possible outputs: the satisfying solution $\sigma$ to $\phi_f$ (L:\ref{ln:returnSAT}), the UNSAT determination of $\phi_f$ (L:\ref{ln:returnUNSAT}), and the bounded-UNSAT determination of $\phi_f$ (L:\ref{ln:returnboundedUNSAT}). 
         $\mainAlg$ returns a satisfying solution $\sigma$ only if (1) $\sigma$ is a solution $\phi_{g}^{\bot}$ (L:\ref{ln:underSAT}) and (2) $\sigma \models T(\psi)$ for every $\psi \in Reqs$ (L:\ref{ln:ReqCheck}). By (1) and Lemma~\ref{lemma:underapprox}, $\sigma$ is a solution to $T(\neg P)\bigwedge_{\psi \in \reqsUnder} T(\psi)$. Together with (2), $\sigma$ is a solution to  $\phi_f$.
         $\mainAlg$ returns UNSAT iff $\phi_g$ is UNSAT (L:\ref{ln:overUNSAT}). By Lemma~\ref{lemma:overapprox}, we show $T(\neg P)\bigwedge_{\psi \in \reqsUnder} T(\psi)$ is UNSAT. Since $\reqsUnder \subseteq Reqs$, the original formula $\phi_f$ is also UNSAT. 
      $\mainAlg$ returns bounded-UNSAT iff the volume of the minimum solution $\sigma_{min}$ to the over-approximated query $\phi_g$ is larger than $\bound$ (L:\ref{ln:exceedsLimit}). Since $\phi_g$ is an over-approximation of the original formula $\phi_f$, any solution $\sigma$ to $\phi_f$ has volume at least $vol(\sigma_{min})$. Therefore, when $vol(\sigma_{min}) > \bound$, $vol(\sigma) > \bound$ for every solution.  Finally, by Thm.~\ref{Thm:tcorrectness}, (1) if $\phi_f$ is satisfiable, then $\neg P \wedge \reqs$ is satisfiable, (2) if $\phi_f$ is UNSAT, then $\neg P \wedge \reqs$ is UNSAT, and (3) if $\phi_f$ does not have a solution with volume not less than $\bound$, then $\neg P \wedge \reqs$ also does not have a solution with volume less than $\bound$ (bounded UNSAT). Therefore, Alg.~\ref{alg:main} is sound for MFTOL bounded satisfiability on inputs $P$, $\reqs$ and $\bound$.
     \qed

\begin{theorem} [Termination] \label{thm:termination}
For an input property $P$, requirements $\reqs$, and a bound $\bound \neq \infty$,   $\mainAlg$ eventually terminates. 
 \end{theorem}
 
\noindent
\textbf{Proof.}
 To prove that $\mainAlg$ always terminates when the input $\bound \neq \infty$,  we need to show that $\mainAlg$  does not get stuck at solving the SMT query via $\textsc{solve}$ (LL:\ref{ln:solve2}-\ref{ln:overUNSAT}),  nor refining $\reqsUnder$ (LL:\ref{ln:ReqCheck}-\ref{ln:ReqCheckFinish}), nor  expanding $\domainUnder$ (LL:\ref{ln:expandingdomain}-\ref{ln:expansionfinish}).
 
 A call to $\textsc{solve}$ (LL:\ref{ln:solve2}-\ref{ln:overUNSAT}) always terminates. By Prop.~\ref{prop:QFfree} both the under- and the over-approximated queries $\phi_g$ and $\phi_{g}^{\bot}$ are quantifier-free. Since the background theory for $P$ is LIA, then $\phi_g$ and $\phi_{g}^{\bot}$ are a quantifier-free LIA formula whose satisfiability is decidable.

 If the requirement checking fails on L:~\ref{ln:ReqCheck},  a violating requirement $lesson$ is added to $\reqsUnder$ (LL:\ref{ln:findviolatingrule}-\ref{ln:ruleRefinement}) which ensures that any future solution $\sigma'$ satisfies $lesson$. Therefore, $lesson$ is never added to $\reqsUnder$ more than once. Given that $\reqs$ is a finite set of \mfotl formulas, at most $|\reqs|$ lessons can be learned before the algorithm terminates.

 The under-approximated domain $\domainUnder$ can be expanded a finite number of times because the size of the minimum solution $vol(\sigma_{min})$ to $\phi_g$  (computed on L:\ref{ln:minisolution}) is monotonically non-decreasing between each iteration of the loop (LL:\ref{ln:loopStart}-\ref{ln:loopEnd}). The size will eventually increase since each relational object in $\domainUnder$ can introduce a finite number of options for adding a new relational object through the grounded encoding of $\phi_g$ on L:\ref{ln:computeUnderground}, e.g., $o.ext \Rightarrow \bigvee_{i=0}^n \exists r_i$. After exploring all options to $\domainUnder$, $vol(\sigma_{min})$ must increase if the algorithm has not already terminated. Therefore, if $\bound \neq \infty$, then eventually $vol(\sigma_{min}) > \bound$, and the algorithm will return bounded-UNSAT instead of expanding $\domainUnder$ indefinitely (LL:\ref{ln:expansionstart}-\ref{ln:expansionfinish}).  \qed

\vskip 0.05in
\noindent
\textbf{Optimality of the solution.}
The following theorem proves that the solution found by $\mainAlg$ has the minimum volume.
 \begin{theorem} [Solution optimality] \label{Thm:optimality}
For a property $P$ and requirements $\reqs$, let $\phi_f$ be the FOL formula  $T(\neg P)\bigwedge_{\psi \in \reqs} T(\psi)$. If $\mainAlg$ finds a solution $\sigma$ for $\phi_f$, then for every $\sigma' \models \phi_f$, $vol(\sigma) \le vol(\sigma')$. 
 \end{theorem}
 
\noindent
\textbf{Proof.}
$\mainAlg$ returns a solution $\sigma$ on L:\ref{ln:returnSAT} only if $\sigma$ is a solution to the under-approximation query $\phi_g^{\bot}$ (computed on L:\ref{ln:computeUnderground}) for some domain $\domainUnder \neq \emptyset$. $\domainUnder$ is last expanded in some previous iterations by adding relational objects to the minimum solution $\sigma_{min}$ (L:\ref{ln:minisolution}) of the over-approximation query $\phi_g'$ (L:\ref{ln:expandingdomain}). Therefore, the returned $\sigma $ has the same number of relational objects as $\sigma_{min}$ ($vol(\sigma_{min}) = vol(\sigma)$). Since $\phi_g$ is an over-approximation of the original formula $\phi_f$, any solution $\sigma'$ to  $\phi_f$ has volume that is at least $vol(\sigma_{min})$. Therefore,  $vol(\sigma) \le vol(\sigma')$. Finally, by Thm.~\ref{Thm:tcorrectness}, the optimal solution of $\neg P \wedge \reqs$ has the same volume as $vol(\sigma)$. 
\qed

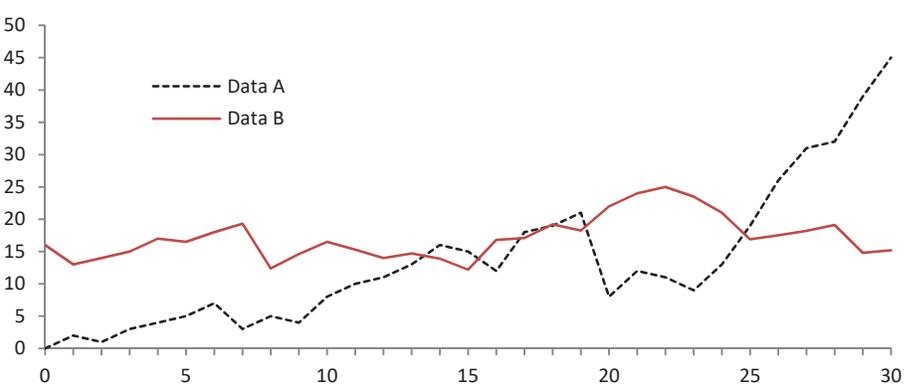
\end{document}